\documentclass[12pt,english]{article}
\usepackage[T1]{fontenc}
\usepackage{kpfonts}
\usepackage{mathpazo}
\usepackage[utf8]{inputenc}
\usepackage{geometry}
\usepackage{babel}
\usepackage{float}
\usepackage{amsmath}
\usepackage{amsthm}
\usepackage{tikz}
\usetikzlibrary{trees,backgrounds,calc,positioning,arrows,patterns} 
\usepackage{istgame}    

\usepackage{tikz}
\usetikzlibrary{trees, positioning}
\usepackage{adjustbox}
\usepackage{amssymb}
\usepackage{enumitem}
\usepackage{graphicx}
\usepackage{setspace}
\usepackage{booktabs}
\usepackage{chngcntr}
\usepackage[authoryear]{natbib}
\usepackage[unicode=true,pdfusetitle,
 bookmarks=true,bookmarksnumbered=false,bookmarksopen=false,
 breaklinks=false,pdfborder={0 0 1},backref=false,colorlinks=false]
 {hyperref}
\hypersetup{
 colorlinks=true, citecolor=blue, linkcolor=blue ,  urlcolor=blue }

\usepackage{titlesec}
\titlespacing*{\paragraph}
{0pt}{1.5ex plus 1ex minus .2ex}{1em}

\usepackage{rotating}

\newtheorem{theorem}{Theorem}
\newtheorem{corol}{Corollary}
\newtheorem{lemma}[theorem]{Lemma}

\newtheorem{assumption}{Assumption}
\newtheorem{definition}{Definition}

\usepackage{longtable}
 \usepackage{tikz}
\usepackage{pgfplots}

\usepackage[flushleft]{threeparttable}
\usepackage[toc,page,header,title]{appendix}
\usepackage{minitoc}

\usepackage{amstext} 
\usepackage{accsupp} 
\usepackage{lipsum}
\usepackage{eurosym}

\DeclareRobustCommand{\officialeuro}{%
  \BeginAccSupp{%
    method=hex,
    unicode,
    ActualText=20AC,
  }%
    \ifmmode\expandafter\text\fi
    {%
      \fontencoding{U}\fontfamily{eurosym}\selectfont e%
    }%
  \EndAccSupp{}%
}

\makeatletter
\theoremstyle{plain}
\newtheorem{prop}{\protect\propositionname}

\usepackage{amsmath}
\DeclareMathOperator*{\argmax}{arg\,max}

\@ifundefined{showcaptionsetup}{}{%
 \PassOptionsToPackage{caption=false}{subfig}}
\usepackage{subfig}
\makeatother

\theoremstyle{plain}

\providecommand{\propositionname}{Proposition}
\providecommand{\corollaryname}{Corollary}

\usepackage{datetime}
\newdateformat{monthyearformat}{
  \monthname[\THEMONTH] \space \THEYEAR
}
\usepackage[normalem]{ulem} 
\usepackage{xcolor}

\newif\ifshowchanges
\showchangestrue            

\definecolor{chgcolor}{RGB}{200,0,0}

\usepackage{cancel} 

\title{Social preferences or moral concerns: \\ What drives rejections in the Ultimatum game?\thanks{We are grateful to Ingela Alger, Ernesto María Gavassa-Pérez, Emin Karagözoğlu, Enrico Mattia-Salonia, Esteban Muñoz-Sobrado and Juan Sebastián Pereyra for their valuable discussions. We thank Pablo Brañas-Garza and Antonio Espín for kindly sharing their data with us. Pau Juan-Bartroli acknowledges funding from the European Research Council (ERC) under the European Union's Horizon 2020 research and innovation programme (grant agreement No 789111 - ERC EvolvingEconomics).}}

\author{Pau Juan-Bartroli\thanks{\color{blue}pau.juanbartroli@tse-fr.eu} \\
	\textit{Toulouse School of Economics}
	\and 
	José Ignacio Rivero-Wildemauwe\thanks{\color{blue}joseignacio.rivero@ucu.edu.uy} \\
	\textit{Universidad Católica del Uruguay}}

\date{}

\begin{document}

\maketitle

\begin{abstract}
Rejections of positive offers in the Ultimatum Game have been attributed to different motivations. We show that a model combining social preferences and moral concerns provides a unifying explanation for these rejections while accounting for additional evidence. Under the preferences considered, a positive degree of spite is a necessary and sufficient condition for rejecting positive offers. This indicates that social preferences, rather than moral concerns, drive rejection behavior. This does not imply that moral concerns do not matter. We show that rejection thresholds increase with individuals’ moral concerns, suggesting that morality acts as an amplifier of social preferences. Using data from \cite{van2023estimating}, we estimate individuals’ social preferences and moral concerns using a finite mixture approach. Consistent with previous evidence, we identify two types of individuals who reject positive offers in the Ultimatum Game, but that differ in their Dictator Game behavior.
\vspace{0.15cm}

{\small \noindent\textbf{JEL codes}: C78, D64, D81, D91

}
\end{abstract}

\newpage

\section{Introduction}

In relevant economic interactions, individuals bargain over how to divide resources. Examples include a buyer and seller negotiating the price of an item, an employer and employee bargaining over wages, and team members determining how to allocate responsibilities. The simplest model to represent these interactions is the Ultimatum Game (UG) (\Citealp{guth1982experimental}). In this game, the \textit{proposer} offers a division of an endowment to the \textit{responder}. If the \textit{responder} accepts this offer, the division is implemented. If the \textit{proposer} rejects this offer, both individuals receive nothing. If both individuals maximize their material payoff, and this is common knowledge, then the unique Subgame Perfect Nash Equilibrium is one where the \textit{responder} accepts any positive offer and the \textit{proposer} offers the lowest possible positive amount. However, the evidence contradicts this prediction, as responders frequently reject positive offers (\Citealp{camerer2003behavioral}; \Citealp{guth2014more}). As discussed below, several studies  propose different explanations of these rejections \citep{levine1998modeling, fehr1999theory, xiao2005emotion, ccelen2017blame, karagozouglu2018time, aina2020frustration}. \par 

In this paper, we show that combining social preferences with moral concerns can account for these behaviors while also explaining other observed patterns. Our analysis also provides insights into how social preferences interact with moral concerns. Understanding how different motives interact in shaping behavior remains an underexplored area of research \citep{gavassa2022moral}. However, recent evidence suggests that prosociality is often driven by a general moral preference to ``do the right thing'' \citep{capraro2018right}, and that deontological motives coexist with outcome-based concerns \citep{chen2022social}. By integrating both dimensions, our framework captures interactions that single-motive models fail to explain.
\par 
Importantly, we do not treat social preferences and moral concerns in isolation. While both concerns have substantial empirical support on their own \citep[e.g.][]{miettinen2020revealed}, evolutionary analysis suggests that it is their \emph{combination} that is favored by natural selection, especially in small-stakes  interactions \citep{alger2020evolution}. This prediction is consistent with recent experimental evidence 
\citep{van2023estimating,alger2024doing}. To the best of our knowledge, this is the first theoretical paper to examine bargaining behavior under the joint presence of social preferences and moral concerns. We show that this interaction is crucial to explain the evidence.

\par

Before discussing our main results, we describe individuals' utility function. Individuals exhibit social preferences when they attach some weight to their opponent's material payoff, which may vary depending on whether the opponent's material payoff exceeds or falls below the individual's own payoff \citep{becker1976altruism,fehr1999theory}.\footnote{In our setting, aversion to disadvantageous inequality is equivalent to spite (i.e., attaching a negative weight to others' material payoff). This term is also referred to as \textit{envy} (\citealp{mui1995economics}; \citealp{levine1998modeling}).} Moral concerns are modeled through universalization reasoning, where individuals act \textit{as if} they assign some weight to the scenario in which others adopt the same strategy as themselves.\footnote{This reasoning can be interpreted in several ways: the individual (i) considers what would happen if all other individuals followed his own course of action; (ii) believes that others are generally similar and will therefore make the same choices (similar to the \textit{false consensus effect}, \citealp{ross1977falseconsensus,dawes1989statistical,butler2015trust}); or (iii) thinks his action influences the likelihood that others will select it as well (similar to \textit{magical thinking}, \citealp{Shafir1992,Daley2017}).}$^,$\footnote{We acknowledge that morality cannot be reduced to universalization reasoning. We focus on this formulation due to its prominence in both theoretical and experimental work. \cite{alger2013homo} and \cite{alger2020evolution} establish its evolutionary foundations, showing that universalization reasoning can emerge as an adaptive trait that promotes cooperation and group cohesion. For alternative moral concepts in economics, see \cite{harsanyi1955cardinal},  \cite{sen1977rational}, \cite{roth2007repugnance}, \cite{roemer2010kantian}, \cite{smith2019humanomics}, \cite{gavassa2022moral} and \cite{drouvelis2024redistribution}.} Under this reasoning, individuals suffer a disutility when selecting a rejection threshold above their own offer. This occurs because in a hypothetical scenario where an individual is matched with someone mimicking their strategy, the resulting payoff is zero—as offers would be rejected.

\par 
Our analysis relies on two features that depart from more standard treatments of the UG (see Section \ref{UG} for more details). First, instead of focusing on computing the game's Nash equilibria, which we relegate to Appendix \ref{AppB}, we consider the strategy that maximizes an individual's utility given his (potentially incorrect) beliefs. We adopt this approach as the former results in (i) multiplicity of equilibria, (ii) no rejections in equilibrium, and (iii) the need to assume that all individuals hold correct beliefs.\footnote{Departing from equilibrium analysis aligns with the approach in \cite{aina2020frustration}, which considers a \textit{subjective rationality} condition in which individuals maximize their utility given their beliefs and preferences.} Second, we consider the UG with \textit{role uncertainty}, where individuals select their offers and rejection thresholds without knowing their realized role and are assigned to each role with equal probability.\footnote{This approach is common in studies examining universalization reasoning in asymmetric games (e.g., \citealp{alger2013homo}; \citealp{munoz2022taxing};  \citealp{salonia2023afoundation}; \citealp{van2023estimating}; \citealp{wildemauwe2023trade}).} This assumption is essential for computing the universalization counterfactual—i.e., ``what if others acted like me?''—and enables us to directly study the relationship between individuals' own offers and their rejection thresholds (see Corollary \ref{C2}).

\par \vspace{0.15cm}
\noindent \textbf{Results:} Our main result characterizes individuals' optimal strategy in the UG as a function of their degrees of spite and morality. We show that the resulting strategy partitions the parameter space into three regions (Proposition \ref{P1}). The optimal strategy gives three main insights. First, the rejections of positive offers in the UG can be rationalized by spite. Specifically, a positive degree of spite is both a necessary and sufficient condition for selecting a positive rejection threshold. In contrast, a positive degree of morality is not even a necessary condition for doing so. This relates to studies showing the spiteful nature of rejections in the UG (\citealp{falk2005driving}; \citealp{herrmann2008antisocial}; \citealp{branas2014fair}).\footnote{ This literature cast doubt on the explanation that individuals’ (prosocial) preferences for fairness motivate the rejection of low offers in the UG (\Citealp{fehr2003nature}; \Citealp{gintis2003explaining}). According to this explanation, individuals who make high transfers in the DG are more likely to select high rejection thresholds.} We argue that the prosocial and spiteful high-punishers documented in \cite{branas2014fair} can be reconciled by considering varying levels of spite and universalization reasoning (see Section \ref{OUS}). Specifically, prosocial punishers exhibit high levels of universalization reasoning and low but positive levels of spite, whereas spiteful punishers exhibit high levels of spite and moderate levels of universalization reasoning (see Table \ref{tab:paujose_core}).

Second, for sufficiently small levels of spite, individuals' offers and rejection thresholds increase in their degree of morality. Thus, although positive rejection thresholds are driven by spite, this does not imply that individuals' degree of morality is irrelevant. Universalization reasoning acts as an \textit{amplifier} of individuals' social preferences, which is a novel observation in the literature. This is consistent with \cite{kimbrough2016norms} and \cite{capraro2022moral} who document a positive correlation between rejection thresholds and measures of morality. Third, our theory predicts a relationship between individuals' offers and their rejection thresholds. If an individual assigns a sufficiently large weight to universalization reasoning, he will select an offer above his rejection threshold. This relates to the results in \cite{ccelen2017blame} and \cite{candelo2019proposer}. \footnote{\cite{ccelen2017blame} document that individuals are less likely to reject offers that are higher than what they would offer as proposers. \cite{candelo2019proposer} find that offers and (hypothetical) rejection thresholds are positively correlated. These authors interpret these findings as evidence of reciprocity and the \textit{false consensus effect}, respectively. We argue that universalization reasoning also accounts for these observations.}
\par \vspace{0.15cm}
\noindent \textbf{Out-of-sample predictions:} While our theoretical model predicts the optimal strategy, identifying which individuals select high rejection thresholds remains an open question. This is because individuals’ rejection thresholds can result from different combinations of parameters. To address this, we use data from \cite{van2023estimating} to compute predicted behavior in both the Dictator Game (DG) and UG, at both the individual and aggregate levels. This approach allows us to examine which types of individuals select positive rejection thresholds in the UG, given their DG behavior, and to infer their underlying preferences.\footnote{We use the DG as it abstracts away from strategic considerations isolating individuals’ motives more clearly. Additionally, it enables us to determine whether high rejection thresholds are selected by individuals who make large transfers in the DG, as predicted by models where punishment is driven by prosocial motives.} In the three-type model, two types of individuals select positive rejection thresholds: those who transfer nothing in the DG and those who transfer a substantial share of the endowment. These types differ in their social preferences but exhibit similar levels of universalization reasoning.\footnote{Our out-of-sample predictions are quantitatively lower than those reported in previous studies (e.g., \citealp{branas2014fair}), with the discrepancy being particularly pronounced for rejection thresholds. We interpret these predictions primarily as a qualitative exercise, given the challenges of comparing estimates from experimental settings that differ from our theoretical framework.} The predicted behavior aligns with the key insights of our model: (i) rejections in the UG are driven by spite, and (ii) universalization reasoning amplifies social preferences. \par \vspace{0.15cm}
\noindent \textbf{Related literature}: This paper contributes to the theoretical and experimental literature on rejection thresholds in the UG (see \cite{camerer2003behavioral} and \cite{guth2014more} for reviews). Previous studies have attributed such rejections to various motivations, including spite (\Citealp{levine1998modeling}), reciprocity (\citealp{ccelen2017blame}), inequity aversion (\Citealp{fehr1999theory}), fairness (\Citealp{karagozouglu2018time}), anger (\Citealp{xiao2005emotion}) and frustration (\Citealp{aina2020frustration}). We contribute to this literature by showing that adding universalization reasoning to a model of social preferences explains several results in a unifying way, including some not accounted for by previous models. More precisely, we explain (i) the rejection of positive offers by subjects who behave either selfishly or prosocially in the DG, (ii) the correlation between rejection thresholds and measures of morality (\citealp{kimbrough2016norms}; \citealp{capraro2022moral}), (iii) why individuals are more likely to reject offers if those are lower than the offer they would have selected as proposers (\citealp{ccelen2017blame}), and (iv) the positive correlation between offers and (hypothetical) rejection thresholds (\citealp{candelo2019proposer}). \par

Finally, this paper contributes to the growing literature using universalization reasoning to explain behavior in several settings. This literature builds on \cite{alger2013homo} and \cite{alger2016evolution} that show the evolutionary stability of universalization reasoning.\footnote{See \cite{alger2023evolutionarily} for a review of the evolutionary foundations of universalization reasoning. For its axiomatic foundations, see \cite{salonia2023afoundation}. For another model of universalization reasoning, see \cite{brekke2003economic}. An alternative way of introducing morality is through the \textit{Kantian equilibrium}, which modifies the equilibrium concept (\Citealp{roemer2010kantian}; \Citealp{roemer2015kantian}). See \cite{dizarlar2023kantian} for the study of bargaining interactions using the Kantian equilibrium.} This reasoning has been used in theoretical studies to rationalize voting (\Citealp{alger2022homo}), paying taxes (\Citealp{munoz2022taxing}), climate change concerns (\Citealp{eichner2021climate}) and behavior in teams (\Citealp{sarkisian2017team}). \cite{alger2017strategic} analyze universalization reasoning and altruism separately across several games. In addition, the mechanism has been successfully tested in laboratory settings (\Citealp{miettinen2020revealed}; \Citealp{juanbartroli2023norms}; \citealp{van2023estimating};  \Citealp{alger2024doing}). Importantly, \cite{van2023estimating} structurally estimate a model combining universalization reasoning and social preferences and find that this model is better at explaining behavior than one considering only one of these concerns separately. This evidence provides an important justification for studying the predictions of a model combining both social preferences and universalization reasoning.\footnote{\cite{alger2024doing} show that making role uncertainty salient makes subjects behave more prosocially. They also structurally estimate the intensity of social preferences and universalization reasoning, finding positive and significant values for the latter.} 

Closer to our work, \cite{juan2024moral}, \cite{rivero2023moral}, and \cite{wildemauwe2023trade} study theoretically the effect of universalization reasoning in bargaining interactions.\footnote{\cite{alger2012homo} characterize the UG's Nash equilibria when individuals have universalization reasoning. We extend their results by computing the Nash equilibria with individuals who have both universalization and social concerns (see Appendix \ref{AppB}). } \cite{juan2024moral} study how the equilibrium of a Divide-the-Dollar game changes with universalization reasoning and compare it with other preferences. \cite{rivero2023moral} and \cite{wildemauwe2023trade} analyze the equilibria of bilateral trade games that feature different information asymmetries in sequential and simultaneous settings with individuals exhibiting universalization reasoning. Besides considering the UG and by allowing for potentially incorrect beliefs, we differ from these studies by examining the interaction between universalization reasoning and social preferences. We show that this is crucial to understanding individuals' behavior and deriving novel theoretical predictions. \par \vspace{0.15cm}
\noindent \textbf{Outline:} The remainder of the paper is organized as follows. Section \ref{S2} introduces the theoretical framework. Section \ref{sec:theory} presents the theoretical results. Section \ref{OUS} considers out-of-sample results. Section \ref{S5} concludes. All the proofs are in Appendix \ref{AppA}.

\section{Theoretical Framework} \label{S2}
\subsection{The Ultimatum Game} \label{UG}
The UG is composed of two players, the \textit{proposer} and the \textit{responder}. The \textit{proposer} offers a division of a pie $x_1 \in [0, w]$ of known size $w > 0$. 
If the \textit{responder} accepts the proposed division, it is implemented: the \textit{proposer} earns $w - x_1$, and the \textit{responder} earns $x_1$. If the \textit{responder} rejects the division, both individuals earn 0. Following previous literature, we assume that responders choose a rejection threshold $x_2 \in [0, w]$ such that any offer below $x_2$ is rejected, and any offer equal to or above $x_2$ is accepted.

To introduce universalization reasoning, we consider the ex-ante symmetric version of the UG, where individuals choose their offer and rejection threshold without knowing their role in the interaction. This defines a symmetric interaction in which \textit{Nature} assigns each player a role with equal probability, and then the players learn their respective roles. A strategy is then a tuple $x = (x_1, x_2) \in [0, w]^2$ where $x_1$ and $x_2$ represent the offer and the rejection threshold selected by the individual in \textit{proposer} and \textit{responder} roles, respectively. The extensive form of the game is presented in Figure \ref{fig:game}.

\begin{figure}[h]
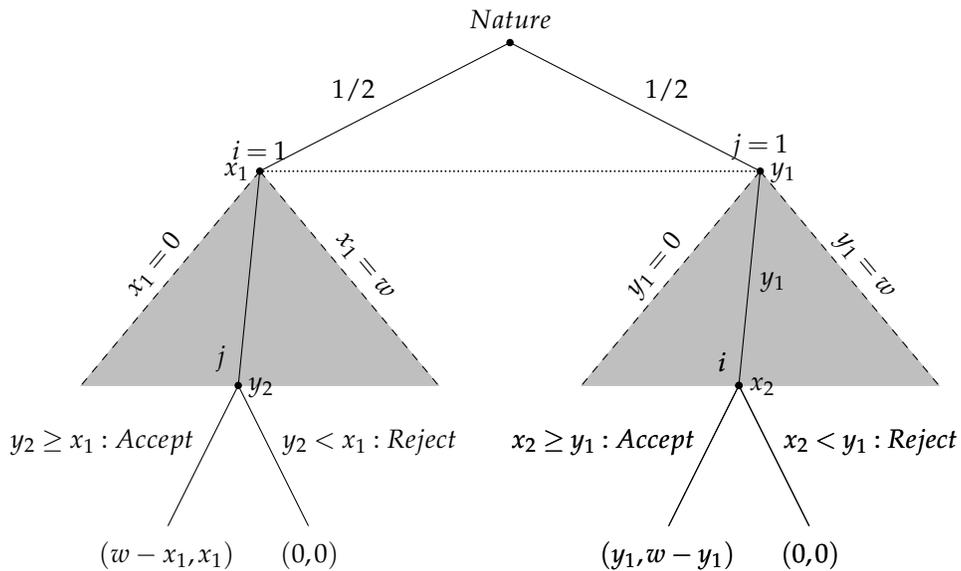

    \caption{Ultimatum Game behind the veil of ignorance}
    \label{fig:game}
    \begin{center}
        \begin{istgame}[scale=1.9,font=\footnotesize] 
            \setistmathTF111 
            \xtdistance{9mm}{35mm} 
    \cntmdistance{15mm}{25mm}{3mm}

            \istroot(0){Nature}
                \istb{1/2}[al] 
                \istb{1/2}[ar]
               
            \endist

            \cntmistb[draw=black,dashed]{x_1=0}[above,sloped]{x_1=w}[above,sloped]
                   \istrootcntm(A)(0-1)
                                      \istb{}[r] 
                   \istbm 
                                      \endist 
       \xtOwner(A){$i=1$}[above]           
       \xtPayoff*(0-1){x_1}[l]            

            \xtdistance{10mm}{10mm} 
            \istroot(B)(A-1)<120>{j}

                \istb{y_2\geq x_1: Accept}[al]{(w-x_1,x_1)} 
                \istb{y_2< x_1: Reject}[ar]{(0,0)} 
            \endist
                        \xtPayoff*(A-1){y_2}[r]

            \cntmistb[draw=black,dashed]{y_1=0}[above,sloped]{y_1=w}[above,sloped]
                   \istrootcntm(C)(0-2)
                   \istb{y_1}[r] 
                   \istbm 
                   \endist 
                    \xtOwner(C){$j=1$}[above]               
       \xtPayoff*(0-2){y_1}[r] 
                   \xtdistance{10mm}{10mm} 
                   
                   \istroot(D)(C-1)<120>{i} 
                   \istb{x_2\geq y_1: Accept}[al]{(y_1,w-y_1)} 
                   \istb{x_2 < y_1: Reject}[ar]{(0,0)} 
                   \endist
            \endist 
\xtPayoff*(C-1){x_2}[r]

\xtInfoset(A)(C){}
        
        \end{istgame}
    \end{center}          
\end{figure}

Examining the ex-ante symmetric version of the game does not rule out applying our results to settings without role uncertainty. Even a proposer who knows they are not a responder may recognize that roles were randomly assigned and could have been reversed. Consequently, individuals who are certain of their role might still contemplate how choices would differ if roles were reversed and others chose their actions. In this sense, role uncertainty makes universalization reasoning more salient relative to social preferences (see \citealp{alger2024doing}).\footnote{ Examining the ex-ante symmetric version of the game admits two interpretations. First, it can be thought of as a mental representation that individuals perform when there is no role uncertainty. For example, the proposer, aware that the roles have been randomly selected, considers how choices would differ if roles were reversed and others chose their actions (\citealp{rawls1971theory}; \citealp{binmore1994game}). Second, it can be interpreted literally, with choices implemented probabilistically and participants unaware of their roles, as in experiments with role uncertainty (\citealp{andreoni2009social}; \citealp{iriberri2011role}). Many real-world markets share this feature, where individuals alternate between roles—for instance, sometimes as buyers and sometimes as sellers—from informal exchanges such as flea markets to online marketplaces.}

\par 
Finally, we do not impose that individuals hold correct beliefs about their opponents' strategies. This is motivated by the one-shot, anonymous nature of the interaction, which makes it difficult for individuals to form accurate beliefs about others’ actions.\footnote{It is widely documented that individuals’ beliefs are heterogeneous and often inaccurate, even in the simplest experimental settings (\citealp{bellemare2008measuring}; \citealp{iriberri2013elicited}).} In our case, individual $i$ holds beliefs $f_{i}(y_1)$ and $f_{i}(y_2)$ on the distributions of offers and rejection thresholds in the population, with $F^{y_1}_{i}(\cdot)$ and $F^{y_2}_{i}(\cdot)$ denoting the corresponding cumulative distribution functions. Both $f_{i}(y_1)$ and $f_{i}(y_2)$ are assumed to be continuous and smooth functions, which may differ from the true distributions $f(y_1)$ and $f(y_2)$. In Appendix B, we extend \cite{alger2012homo} to compute the Nash equilibrium of the game when individuals have social preferences and universalization reasoning. The key distinction between computing the strategy that maximizes individuals' utility given their beliefs and deriving the Nash equilibrium of the game is that, in the former, individuals are not required to hold correct beliefs. This assumption allows for rejections of some positive offers as a best response and simplifies the comparative static analysis.
\subsection{Preferences}
Consider an interaction between two individuals and let \textit{X} denote the common set of pure strategies, assumed to be non-empty, convex, and compact. Let $\pi: [0, w]^2 \rightarrow \mathbb{R}$ be the individuals' material payoff function, which represents how they evaluate monetary outcomes. Specifically, $\pi(x_i, x_j) \equiv v({\widetilde\pi(x_i, x_j)})$, where $\widetilde\pi(x_i, x_j)$ is individual $i's$ monetary payoff under strategy profile $(x_i, x_j)$, and $v: \mathbb{R} \rightarrow \mathbb{R}$ captures how individuals evaluate monetary payments. We assume that individuals' utility functions combine social preferences and universalization reasoning (\Citealp{alger2020evolution}; \citealp{van2023estimating}). More concretely, individual $i$'s utility function is:
\begin{eqnarray} \label{1} 
u_{i}(x_i, x_j) &=&
(1 - \kappa_i)\pi(x_i, x_j)
\\
&-&  \alpha_i \max[\pi(x_j, x_i) - \pi(x_i, x_j), 0] \nonumber \\
&-& \beta_i \max[\pi(x_i, x_j) - \pi(x_j, x_i), 0] \nonumber \\
&+& \kappa_i \pi(x_i, x_i), \nonumber
\end{eqnarray}
where $x_i$ (resp. $x_j$) is the strategy of individual $i$ (resp. $j$), and $\pi(x_i, x_j)$ (resp. $\pi(x_j, x_i)$) is the material payoff of individual $i$ (resp. $j$) under the strategy profile $(x_i, x_j)$.

The utility function has three parameters. Two of these parameters capture inequity aversion (as in \Citealp{fehr1999theory}). The parameter $\alpha_i$ represents $i$'s (dis)utility from disadvantageous inequality (i.e., obtaining a payoff lower than individual $j$). Similarly, $\beta_i$ represents individual $i$'s (dis)utility from advantageous inequality (i.e., obtaining a payoff higher than individual $j$). Finally, interpreting the last term in (\ref{1}) as a representation of Kant's categorical imperative (\Citealp{kant1785categorical}), “act only on the maxim that you would at the same time will to be a universal law”, we refer to $\kappa_i \in [0, 1]$ as $i$'s degree of morality. In this line, $\pi(x_i, x_i)$ measures individual $i$'s material payoff if individual $j$ used the same strategy as individual $i$, while $\kappa_i$ measures the importance he attaches to this counterfactual scenario.\footnote{Alternatively, the parameter $\kappa_i$ can be interpreted in two ways: (i) as the probability that an individual is paired with someone who behaves like them, similar to the \textit{false consensus effect} \citep{ross1977falseconsensus,dawes1989statistical,butler2015trust}; or (ii) as a proxy for \textit{magical thinking}, wherein an individual believes that choosing $x$ makes it more likely that others choose $x$ as well \citep{Shafir1992,Daley2017}.}

Utility function (\ref{1}) nests several preferences as special cases. For example, setting $\alpha_i = \beta_i = \kappa_i = 0$ represents self-interest. Altruistic preferences (\Citealp{becker1976altruism}) can be achieved by setting $\kappa_i = 0$ and $\alpha_i = -\beta_i$ for a $\beta_i \in (0, 1)$. Inequity aversion preferences (\Citealp{fehr1999theory}) can be attained by setting $\kappa_i = 0$ and $\alpha_i \geq \beta_i > 0$. Spiteful preferences (\Citealp{levine1998modeling}) can be realized by setting $\kappa_i = 0$ and $\alpha_i = -\beta_i$ for a $\beta_i \in (-1, 0)$. Lastly, \textit{homo moralis} preferences (\Citealp{alger2013homo}) can be obtained by setting $\alpha_i = \beta_i = 0$ and $\kappa_i \in (0, 1]$. Note that (\ref{1}) does not directly consider reciprocity motives (\citealp{rabin1993incorporating}; \citealp{dufwenberg2004theory}). Although reciprocity is relevant in the UG, we are interested in studying the relationship between behavior in the UG and the DG, for which reciprocity is silent. \par

In this setting, \cite{van2023estimating} is particularly relevant, as the authors conduct a lab experiment to structurally estimate  (\ref{1}). They find that the estimated parameters at the aggregate level satisfy $\hat{\alpha} > 0$, $\hat{\beta} = 0$ and $\hat{\kappa} > 0$, suggesting that the representative agent dislikes disadvantageous inequality, is indifferent with respect to advantageous inequality, and has a positive degree of morality.\footnote{It is worth noting  that \cite{van2023estimating} find $\beta_i \neq 0$ and $\alpha_i < 0$ for some individuals when considering estimates at the individual level. In addition, \cite{alger2024doing} find $\beta \neq 0$ at the aggregate level.}  This motivates our first assumption. 
\begin{assumption} \label{Types}
   For any $i$, $\beta_i = 0$, $\alpha_i \geq 0$ and $\kappa_i \in [0, 1]$. 
\end{assumption}
Thus, we assume that individuals are unaffected by advantageous inequality but may be concerned with disadvantageous inequality and assign some weight to universalization reasoning. Our results also extend when $\beta_i \neq 0$, and $\alpha_i < 0$. However, allowing for these cases complicates the exposition of the results without providing additional insight. Footnote (\ref{foot2}) clarifies how the optimal strategy changes when $\beta_i \neq 0$, while footnote (\ref{foot1}) addresses the case where $\alpha_i < 0$.

\section{Theoretical Results}\label{sec:theory}

In this section, we derive the optimal strategy of an individual whose utility function combines social preferences and moral concerns, as in equation (\ref{1}). Section \ref{subsec:prelim} introduces Lemmas \ref{L1} and \ref{L2}, which facilitate this characterization. Section \ref{subsec:characterization} characterizes the optimal strategy. Section \ref{subsec:relation} presents several corollaries and discusses how they contribute to rationalizing evidence from the experimental literature.

\subsection{Preliminary Analysis}\label{subsec:prelim}
We begin by introducing Assumptions \ref{beliefs} and \ref{concave}, which facilitate the characterization of the optimal strategy. Assumption \ref{beliefs} restricts individuals' beliefs.

\begin{assumption} \label{beliefs}
 For any $i$, $F^{y_2}_{i}(\frac{w}{2}) = F^{y_1}_{i}(\frac{w}{2}) = 1$.   
\end{assumption}

Assumption \ref{beliefs} implies that all individuals believe others' offers and rejection thresholds do not exceed $\frac{w}{2}$, consistent with previous evidence  (\Citealp{camerer2003behavioral}; \Citealp{guth2014more}). In practice, this assumption restricts individuals' optimal strategies to offers and rejection thresholds not exceeding $\frac{w}{2}$. Therefore, spite (i.e., behindness aversion) influences individuals' utility only in the \textit{responder} role, as proposers always receive a payoff at least equal to that of responders. 

Assumption \ref{concave} imposes structure on the function individuals use to evaluate monetary payoffs resulting from the interaction.
\begin{assumption} \label{concave}
Function $v$ is increasing, strictly concave and differentiable.
\end{assumption}

Assumption \ref{concave} implies that individuals exhibit decreasing marginal utility of money and are risk-averse behind the veil of ignorance.\footnote{This assumption is also imposed in \cite{alger2013homo} when analyzing the DG and UG with individuals who employ universalization reasoning. Our results remain robust when assuming a linear function, $v(x) = x$. However, as explained in footnote (\ref{Explanation}), this assumption is less convenient when analyzing the DG.} Under Assumptions \ref{beliefs} and \ref{concave}, and normalizing $v(0) = 0$, individual $i$ maximizes the following utility function: 
\begin{eqnarray} \label{eq: (5)} 
{u}_{i}(x_1, x_2)&=&
(1 - \kappa_i)v(w - x_1)F^{i}_{y_2}(x_1) \\
&+& \int_{x_2}^{\frac{w}{2}} [(1 - \kappa_i)v(y_1) - \alpha_i (v(w - y_1) - v(y_1))]dF_{i}^{y_1} \nonumber \\
&+&  \kappa_i \mathbf{1}_{\{x_1 \geq x_2\}}[v(w - x_1) + v(x_1)].
\nonumber 
\end{eqnarray} \par 
\begin{assumption} \label{A4}
 Function $u_i$ is is strictly concave on each region $x_1 < x_2$ and $x_1 > x_2$, continuous on the whole domain and differentiable (except possibly on the diagonal $x_1 = x_2$).     
\end{assumption}

\noindent Assumption \ref{A4} is a technical assumption that makes the problem concave and well-defined.\footnote{Concavity is assured if the belief distributions $F_{y_1}^i$ and $F_{y_2}^i$ are absolutely continuous on $[0, w/2]$ with densities $f_1, f_2$ that do not increase too fast. Formally, there exist finite constants $b_1, b_2 > 0$ such that
$\sup_{x \in [0,w/2]} \big| f_1'(x) \big| \le b_1
\quad$ and $\sup_{x \in [0,w/2]} \big| f_2'(x) \big| \le b_2.$ The benchmark uniform case $F = \mathcal{U}[0, w/2]$ satisfies these conditions with $b_1 = b_2 = 0$. No specific numerical values of $b_1, b_2$ are required for the results. \label{Convavity}
} Individual $i$'s utility comprises three terms:\footnote{When $\beta_i > 0$, (\ref{eq: (5)}) would have an additional term equal to $-\beta \frac{1}{2}[v(w - x_1) - v(x_1)]F^{i}_{y_2}(x_1)$ capturing the disutility of being ahead of the other player when acting as the \textit{proposer}. This would make $x_s$, $\widetilde{x}_{1}(\kappa_i)$, and $\hat{x}(\kappa_i, \alpha_i)$ from Definitions \ref{D1}, \ref{D2}, and \ref{D4} below to depend positively on $\beta_i$, but not affect the qualitative patterns described in the main text.\label{foot2}} 
\begin{itemize}[topsep=0pt, partopsep=0pt, parsep=0pt, itemsep=0pt]
    \item The first term represents individual $i$'s payoff when offering $x_1$ as \textit{proposer} given his expected distribution of rejection thresholds. Increasing $x_1$ decreases individual $i's$ payoff when his offer is accepted but increases the likelihood of this occurring. \vspace{0.15cm}
    \item The second term represents $i$'s (net) payoff when setting a rejection threshold of $x_2$ as \textit{responder} given his expected distribution of offers. Increasing $x_2$ decreases the likelihood of accepting an offer and increases the lowest offer accepted.\vspace{0.15cm}
    \item The third term captures individual $i$'s hypothetical payoff if others adopted the same strategy as himself, as implied by the universalization reasoning. That is, the payoff individual $i$ would receive if matched with someone choosing strategy $x = (x_1, x_2)$. The universalization concern implies that individual $i$ is worse off when (i) rejecting an offer he would propose himself (i.e., $x_1 < x_2$), and (ii) offering a division below the equal split (i.e., $x_1 < \frac{w}{2}$). 
\end{itemize} 
To characterize individuals' optimal strategy $x^* = ({x_1}^*, {x_2}^*)$, it is useful to define the following four objects. 
\begin{definition} \label{D1}
(\textit{Selfish Offer}): $x_{s} \equiv \argmax_{x \in [0, w]}  \frac{1}{2}v(w - x_1)F^{i}_{y_2}(x_1)$. \end{definition} 
\begin{definition} \label{D2}
(\textit{}{Constrained optimal offer}): $\widetilde{x}_{1}(\kappa_i) \equiv \argmax_{x_1 \in [0, w]} \frac{1}{2}(1 - \kappa_i) v(w - x_1)F^{i}_{y_2}(x_1) + \frac{1}{2}\kappa_i[v(w - x_1) + v(x_1)]$.
\end{definition}
\begin{definition} \label{D3}
(\textit{Constrained optimal rejection threshold}): $\underline{x_2}(\kappa_i, \alpha_i) \equiv  \argmax_{x_2 \in [0, w]} \int_{x_2}^{\frac{w}{2}} [(1 - \kappa_i)v(y_1) - \alpha_i (v(w - y_1) - v(y_1))]dF_{i}^{y_1}$. 
\end{definition}
\begin{definition} \label{D4}
    (\textit{Symmetric optimal strategy}): $\hat{x}(\kappa_i, \alpha_i) \equiv \argmax_{x \in [0, w]}u_{i}(x, x)$.
\end{definition}

The \textit{Selfish offer} is the offer that maximizes individual $i$'s material payoff, given their beliefs about others' rejection thresholds. This is the offer selected by someone with $\kappa_i = \alpha_i = 0$. The \textit{Constrained optimal offer} and the \textit{Constrained optimal rejection threshold} are the offer and the rejection threshold that maximizes individual $i$'s utility, conditional on the offer exceeding his rejection threshold. 

Finally, the \textit{Symmetric optimal strategy} is the strategy that maximizes individual $i$'s utility, subject to the constraint that the offer and rejection threshold are equal. \par 

To simplify notation, we omit the parameters of the above-defined objects and refer to them simply as $x_s$, $\widetilde{x}_{1}$, $\underline{x_2}$, and $\hat{x}$. Lemma \ref{L1} characterizes the optimal strategy $x^*$ as a function of $x_s$, $\widetilde{x}_{1}$, $\underline{x_2}$ and $\hat{x}$, showing that it can be divided into three regions.
\begin{lemma} \label{L1}
    Let $x^* = ({x_1}^*, {x_2}^*)$ denote the strategy $(x_1, x_2) \in [0, w]^2$ that maximizes (\ref{eq: (5)}). Then, 
\begin{itemize}
    \item If $\widetilde{x}_{1} \geq \underline{x_2}$, then $x_1^* = \widetilde{x}_{1}$  and $x_2^* = \underline{x_2}$.
    \item If $\widetilde{x}_{1} < \underline{x_2}$,  then
    \begin{itemize}
    \item If $u(x_s, \underline{x_2}) \geq u(\hat{x}, \hat{x})$, $x_1^* = x_s$  and $x_2^* = \underline{x_2}$.
    \item If $u(x_s, \underline{x_2}) < u(\hat{x}, \hat{x})$, $x_1^* = x_2^* = \hat{x} \in [\widetilde{x}_{1}, \underline{x_2}]$.
    \end{itemize}
\end{itemize} 
\end{lemma}
\begin{proof}
All proofs are in Appendix \ref{AppA}.
\end{proof}
Intuitively, the individual computes $\widetilde{x}_{1}$ and $\underline{x_2}$ separately. If $\widetilde{x}_{1} \geq \underline{x_2}$, then, by Definitions \ref{D2} and \ref{D3}, it is optimal to set $x_1^* = \widetilde{x}_{1}$ and $x_2^* = \underline{x_2}$. When $\widetilde{x}_{1} < \underline{x_2}$, setting $x_1^* = \widetilde{x}_{1}$ and $x_2^* = \underline{x_2}$ is not optimal given the disutility the individual suffers when setting $x_1^* < x_2^*$. In that case, the individual selects either (i) $x_1^* = x_2^* = \hat{x}$ or (ii) $x_1^* = x_s$ and $x_2^* = \underline{x_2}$ (with $x_1^* < x_2^*$). Lemma \ref{L2} defines two functions $\tilde{\alpha}(\kappa_i)$ and $\tilde{\kappa}(\alpha_i)$, which characterize the relationship between $\widetilde{x}_{1}$ and $\underline{x_2}$, and $u(x_s, \underline{x_2})$ and $u(\hat{x}, \hat{x})$, respectively.
\begin{lemma} \textit{} \label{L2}
\begin{itemize}[topsep=0pt, partopsep=0pt, parsep=0pt, itemsep=0pt] 
    \item For any $\kappa_i \in [0, 1)$, if $\alpha_i = \tilde{\alpha}(\kappa_i) \equiv (1 - \kappa_i)\frac{v(\tilde{x}(\kappa_i))}{v(w - \tilde{x}(\kappa_i))- v(\tilde{x}(\kappa_i))}$, then $\widetilde{x}_{1} = \underline{x_2}$. \vspace{0.15cm}
    \item For any $\alpha_i > 0$,  there exists  
$\tilde{\kappa}(\alpha_i)$, such that, if $\kappa_i = \tilde{\kappa}(\alpha_i)$, then $u(\hat{x}, \hat{x}) = u(x_s, \underline{x_2})$.
\end{itemize}

\end{lemma}

Two key remarks are worth emphasizing. First, there exists a unique threshold $\overline{\alpha} \equiv \frac{v(x_s)}{v(w - x_s) - v(x_s)} > 0$ such that $\tilde{\alpha}(0) = \overline{\alpha}$ and $\tilde{\kappa}(\overline{\alpha}) = 0$. Therefore, in the ($\alpha$, $\kappa$) space, $\tilde{\alpha}(\kappa_i)$ and $\tilde{\kappa}(\alpha_i)$ intersect at ($\overline{\alpha}$, $0$) (see Figure \ref{F1}). Second, as shown in Proposition \ref{P1}, $\tilde{\alpha}(\kappa_i)$  is relevant for computing the optimal strategy only when $\alpha > \overline{\alpha}$, whereas $\tilde{\kappa}(\alpha_i)$ is relevant only when $\alpha < \overline{\alpha}$. In these regions, $\tilde{\alpha}(\kappa_i)$ is generally decreasing in $\kappa_i$, while $\tilde{\kappa}(\alpha_i)$ is generally increasing in $\alpha_i$.\footnote{Although we do not show these comparative static analysis for every possible specification, we find evidence supporting them in most reasonable functional forms.}  

\subsection{The optimal strategy}\label{subsec:characterization}

\begin{prop} \label{P1}
Let $\overline{\alpha} \equiv \frac{v(x_s)}{v(w - x_s) - v(x_s)} > 0$. Let $x^* = ({x_1}^*, {x_2}^*)$ denote the strategy $(x_1, x_2) \in [0, w]^2$ that maximizes (\ref{eq: (5)}). Then, 
\begin{itemize}
\item If $\alpha_i \leq \tilde{\alpha}(\kappa_i)$, $x_1^* = \widetilde{x}_{1}$ and ${x_2}^* = \underline{x_2}$ (with $x_1^* \geq x_2^*$).
\item If either (i) $\alpha_i \in (\tilde{\alpha}(\kappa_i), \overline{\alpha})$ or (ii) $\alpha_i > \overline{\alpha}$ and $\kappa_i > \tilde{\kappa}(\alpha_i)$, $x_1^* = x_2^*  = \hat{x}$. 
\item If $\alpha_i > \overline{\alpha}$ and $\kappa_i \leq \tilde{\kappa}(\alpha_i)$, $x_1^* = x_s$ and ${x_2}^* = \underline{x_2}$ (with $x_1^* < x_2^*$).

\end{itemize}     
\end{prop}

Proposition \ref{P1} shows that $x^*$ can be divided into three regions depending on the relationship between $x_1^*$ and $x_2^*$. Figure \ref{F1} illustrates these three regions for the case when $w = 10$, $v(x) = \frac{x^{(1 - \rho)}}{1 - \rho}$ with $\rho = 0.05$, and beliefs on offers and rejection thresholds follow a beta distribution normalized in $[0,\frac{w}{2}]$ 
(see Figure \ref{Fig8}). In Region 1, the optimal strategy is $x_1^* = \widetilde{x}_{1}$ and ${x_2}^* = \underline{x_2}$, and occurs when $\alpha_i \leq \tilde{\alpha}(\kappa_i)$. In Region 2, the optimal strategy is $x_1^* = {x_2}^* = \hat{x}$, and occurs when either (i) $\alpha_i \in (\tilde{\alpha}(\kappa_i), \overline{\alpha})$ or (ii) $\alpha_i > \overline{\alpha}$ and $\kappa_i > \tilde{\kappa}(\alpha_i)$. Finally, in Region 3, the optimal strategy is $x_1^* = x_s$ and ${x_2}^* = \underline{x_2}$, and occurs when $\alpha_i > \overline{\alpha}$ and $\kappa_i < \tilde{\kappa}(\alpha_i)$.

\begin{figure}[H]
    \begin{center}
\includegraphics[width= 0.85\textwidth]{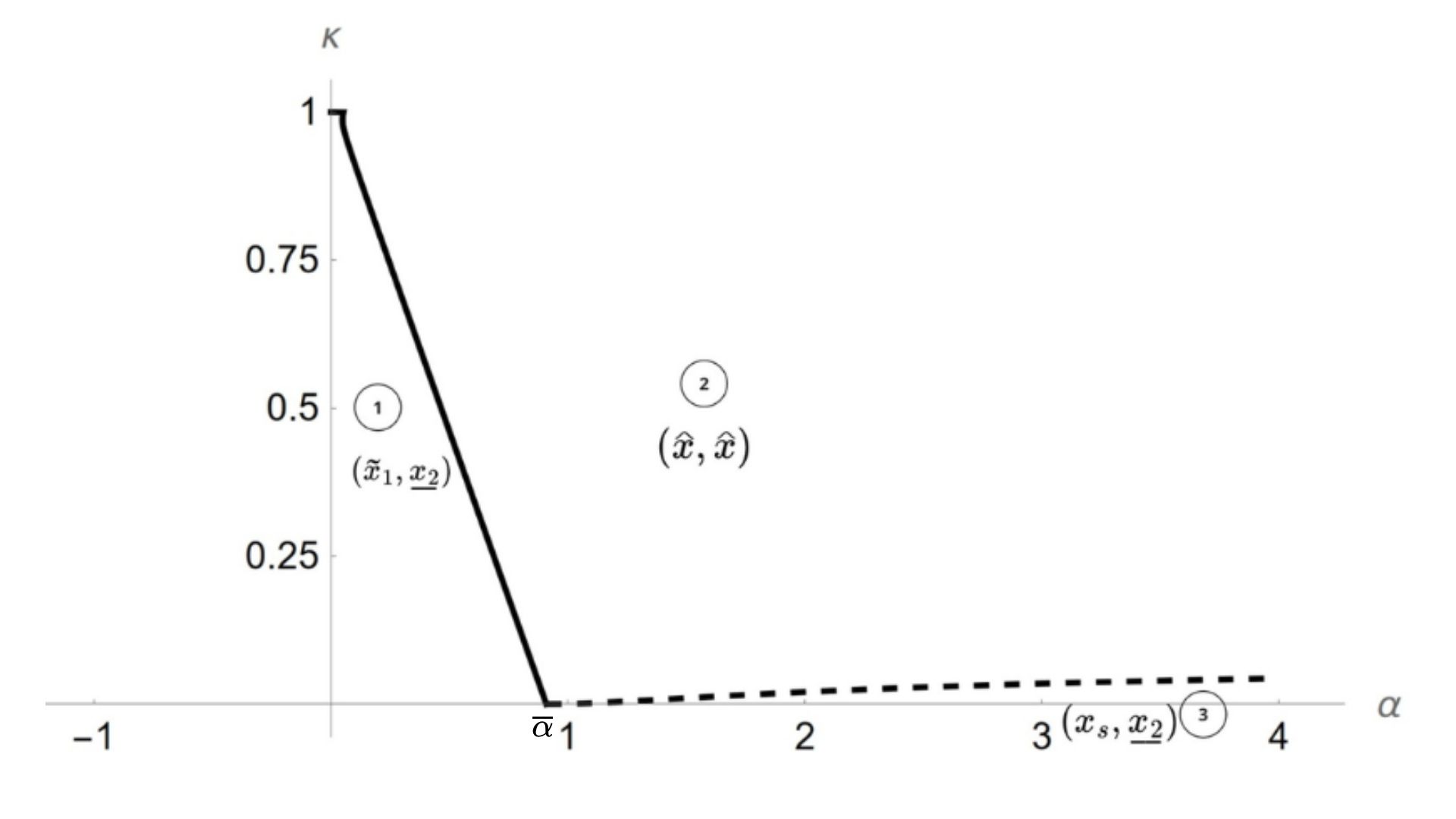}
    \end{center}\vspace{-0.4ex}
    \caption{Optimal strategy when $w = 10$, $v(x) = \frac{x^{(1 - \rho)}}{1 - \rho}$ with $\rho = 0.05$, and $F^{i}_{y_2}(x) = F^{i}_{y_1}(x) = F_{[0,w/2]}(2, 4)$.}
    \label{F1}
\end{figure}
\noindent Before proceeding with detailing the strategy of an individual with $\alpha_i > 0$ and $\kappa_i > 0$, it is useful to briefly describe three polar cases: 
\begin{itemize}
    \item \textbf{Selfish individual} ($\alpha_i=\kappa_i=0$): This individual selects $x_1^* = x_{s}$ and ${x_2^*} = 0$. This implies that the individual never rejects a positive offer and selects the offer that maximizes his expected payoff given his beliefs about others’ rejection behavior. 
    \item \textbf{Spiteful individual} ($\alpha_i > 0$ and $\kappa_i=0$): This individual selects ${x_1^*} = x_s$ and ${x_2^*} = \underline{x_2} > 0$. This implies that the individual selects a strictly positive rejection threshold, which increases with $\alpha_i$, and the payoff-maximizing offer for any value of $\alpha_i$.
    \item \textbf{Homo moralis individual} ($\alpha_i = 0$ and $\kappa_i \in (0, 1)$): This individual selects $x_1^* = \widetilde{x}(\kappa_i)$ and ${x_2^*} = 0$. This individual never rejects a positive offer, while his offer increases with $\kappa_i$.
\end{itemize}

\par

\noindent When the individual combines both $\alpha_i > 0$ and $\kappa_i \in (0, 1]$, his strategy is as follows. If the individual is not sufficiently spiteful, i.e., $\alpha_i \leq \overline{\alpha}$, he faces a trade-off between (i) selecting the strategy with the offer and rejection threshold that maximizes constrained utility (i.e., $\widetilde{x}_{1}$ and $\underline{x_2}$) and (ii) adopting a symmetric strategy that accepts offers yielding a negative payoff (i.e., $x_2 <  \underline{x_2}$). The resolution of this trade-off depends on $\alpha_i$. When $\alpha_i < \tilde{\alpha}(\kappa_i)$, then $\widetilde{x}_{1} > \underline{x_2}$, which implies that setting $x_1 = \widetilde{x}_{1}$ and $x_2 = \underline{x_2}$ is optimal (see Definitions \ref{D2} and \ref{D3}). On the other hand, when $\alpha_i \in (\tilde{\alpha}(\kappa_i), \overline{\alpha})$, then $\widetilde{x}_{1} < \underline{x_2}$, which implies that setting $x_1 = \widetilde{x}_{1}$ and $x_2 = \underline{x_2}$ is not optimal, as individual $i$ suffers a disutility when rejecting offers above what she would have offered. In that case, individual $i$ prefers adopting a symmetric strategy with $x_1^* = x_2^* = \hat{x} \in (\widetilde{x}_{1}, \underline{x_2})$ where the additional disutility is absent.   \par 

When the individual is sufficiently spiteful (i.e., $\alpha_i > \overline{\alpha}$), he faces a trade-off between (i) adopting a symmetric strategy with $x_2 < \underline{x_2}$ and (ii) selecting a rejection threshold exceeding their offer in the \textit{proposer} role (i.e., $x_1 < x_2$). The resolution of this trade-off depends on $\kappa_i$. To see this, note that all strategies with $x_1 < x_2$ give a payoff of zero if the strategy is universalized (i.e., the third term in (\ref{eq: (5)}) is zero). When $\kappa_i$ is low, the individual selects the strategy that maximizes his (net) payoff (i.e., $x_1^* = x_s$ and $x_2^* = \underline{x_2}$) as the disutility of setting $x_1 < x_2$ is low.\footnote{Conditional on selecting a strategy with $x_1 < x_2$, the individual prefers to set $x_1 = x_s$ as this maximizes his material payoff when proposer.} Conversely, when $\kappa_i$ is high, the individual chooses an offer equal to their rejection threshold (i.e.,  $x_1^* = x_2^*$), accepting offers that yield a negative (net) payoff (i.e., ${x_1}^* = {x_2}^* \in [\widetilde{x}_{1}, \underline{x_2}]$).  \par

Although predictions in all three regions described in Proposition \ref{P1} are theoretically possible, the area corresponding to Region 3 appears quantitatively negligible under most reasonable specifications. Moreover, for reasonable values of $\alpha_i$ and $\kappa_i$, most individuals fall within Region 1. To illustrate this, we use the individual estimates of $\alpha_i$ and $\kappa_i$ from  \cite{van2023estimating} to classify subjects into the three regions of Figure \ref{F1}. Figure \ref{fig:dots} (in Appendix \ref{AppC}) illustrates that 95\% of subjects (91 out of 96) are situated in Region 1 while 5\% (5 out of 96) are in Region 2. This classification is quantitatively similar when using other functional forms of utility and beliefs. \par

\subsection{Relation with experimental results}\label{subsec:relation}
In this section, we relate the optimal strategy characterized in Proposition \ref{P1} to findings from the experimental literature. A key question is what defines individuals who reject positive offers. Corollary \ref{C1} shows that punishments observed in the UG have a spiteful nature.
\begin{corol}\label{C1}
For any $\kappa_i \in [0, 1)$, $\alpha_i \leq 0$ is a necessary and sufficient condition for ${x_2}^* = 0$.     
\end{corol}

Therefore, a necessary and sufficient condition for rejecting positive offers is having $\alpha_i > 0$.\footnote{Note that ${x_2}^* = 0$ for any $\alpha_i \leq 0$. Thus, the optimal strategy of an individual with $\alpha_i < 0$ coincides with that of an individual with $\alpha_i = 0$. However, when $\kappa_i = 1$ and $\alpha_i = 0$, the individual disregards material payoff and focuses solely on the implications of universalizing their strategy. In this case, ${x_1}^* = \frac{w}{2}$ and ${x_2}^* \in [0, \frac{w}{2}]$, implying that ${x_2}^*$ may be strictly positive. \label{foot1}} On the other hand, having $\kappa_i > 0$ is not a necessary condition for selecting a rejection threshold above zero. Intuitively, $\kappa_i$ acts as an amplifier of $\alpha_i$ and is therefore relevant only when $\alpha_i > 0$. 
Corollary \ref{C1} relates to the findings documented in  \cite{branas2014fair}, who identify two types of high-punishers: one prosocial type, characterized by high offers in the DG, and one spiteful type, characterized by zero transfers. These results are consistent with two types of individuals: one with high $\kappa_i$ and low (but positive) $\alpha_i$, and one with low $\kappa_i$ and high $\alpha_i$. Although both types select high rejection thresholds in the UG, the former transfers large amounts in the DG, whereas the latter transfers zero. We confirm this argument in Section \ref{OUS}, where we conduct an out-of-sample exercise with data from \cite{van2023estimating}. 

\par 
Although $\widetilde{x}_{1}$, $\hat{x}$, and $\underline{x_2}$ are continuous and increasing in $\kappa_i$, this is not always the case for $x_1^*$ and $x_2^*$. Intuitively, $x_1^*$ and $x_2^*$ increase with $\kappa_i$ when they fall within one of the three regions defined in Proposition \ref{P1}. However, $x_1^*$ (resp. $x_2^*$) jumps up (resp. down) when changing from Region 3 to Region 2 (see Figure \ref{F3}).Corollary \ref{C3} shows that having $\alpha_i \leq \overline{\alpha}$ is sufficient to have both $x_1^*$ and $x_2^*$ continuous and increasing in $\kappa_i$. 
 \par
\begin{corol} \label{C3}
   When $\alpha_i \in [0, \overline{\alpha})$, ${x_1}^*$ and ${x_2}^*$ are continuous in $\kappa_i$ with $\frac{\partial {x_1}^*}{\partial \kappa_i} \geq 0$ and $\frac{\partial {x_2}^*}{\partial \kappa_i} \geq 0$. 

\end{corol}

As mentioned before, $\overline{\alpha}$ appears quantitatively large under most reasonable specifications, suggesting that this condition is met by the majority of individuals. Therefore, Corollary \ref{C3} aligns with experimental studies that show a positive correlation between rejection thresholds and measures of morality. \cite{kimbrough2016norms} document a positive correlation between individuals' rejection thresholds and their degree of norm-following.\footnote{\cite{juanbartroli2023norms} show that models of norm compliance can be reformulated as models of \textit{homo moralis} preferences, implying a positive correlation between an individual's degree of norm compliance ($\gamma_i$) and their degree of morality ($\kappa_i$). Furthermore, \cite{juanbartroli2023norms} documents this positive correlation.} Similarly, \cite{capraro2022moral} document a positive correlation between individuals' rejection thresholds and two measures of morality (based on a ``trade-off'' game and a ``moral foundations'' questionnaire).\footnote{This complementarity between morality and social preferences relates to the research in psychology showing that individuals with stronger moral identities are more likely to engage in costly prosocial behavior (e.g., \citealp{aquino2002self}).} 
\par 
Figure \ref{F3} illustrates how the optimal offer and rejection threshold vary with $\kappa_i$ when (i) $\alpha_i = 0$ (dashed blue), (ii) $\alpha_i = 0.5$ (dotted black), and (iii) $\alpha_i = 3$ (dashed red). These scenarios represent all possible cases: (i) the optimal strategy is always in Region 1 ($\alpha_i = 0$), (ii) it moves from Region 1 to Region 2 ($\alpha_i = 0.5$), and (iii) it moves from Region 3 to Region 2 ($\alpha_i = 3$). \par 
\begin{figure}[H]
  \centering
  \begin{minipage}[b]{0.45\textwidth}
    \includegraphics[width=\textwidth]{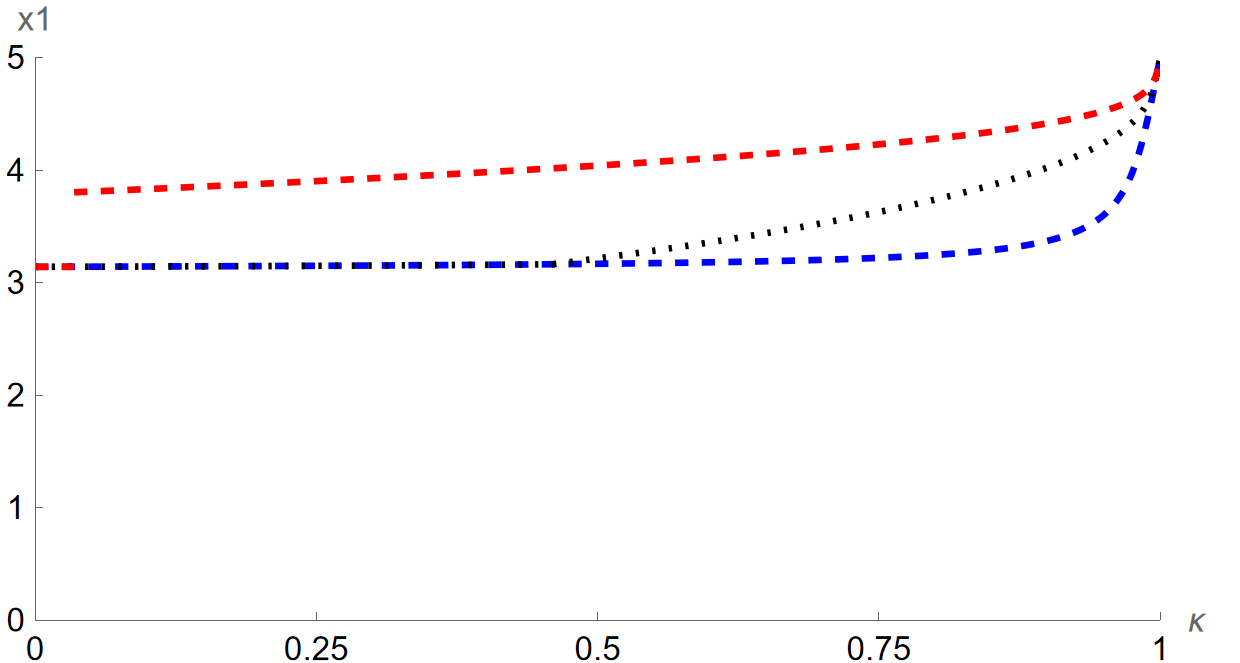}
    \captionsetup{labelformat=empty}
    \caption{Optimal offer }
    \addtocounter{figure}{-1} 
  \end{minipage}
  \hfill
  \begin{minipage}[b]{0.45\textwidth}
    \includegraphics[width=\textwidth]{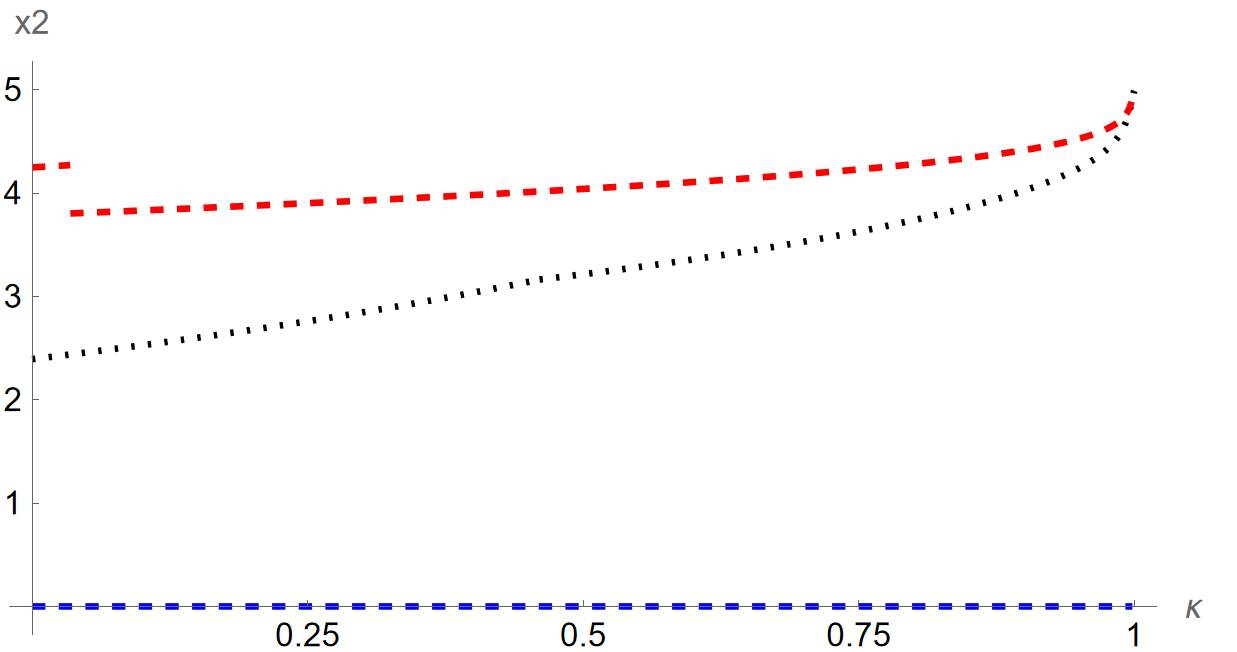}
    \captionsetup{labelformat=empty}
    \caption{Optimal rejection threshold }
    \addtocounter{figure}{-1} 
  \end{minipage}
  \caption{Optimal offer (left) and optimal rejection threshold (right) when $w = 10$, $v(x) = \frac{x^{(1 - \rho)}}{1 - \rho}$ with $\rho = 0.05$, and $F^{i}_{y_2}(x) = F^{i}_{y_1}(x) = F_{[0,w/2]}(2, 4)$. $\alpha_i = 0$ (dashed blue), $\alpha_i = 0.5$ (dotted black), and $\alpha_i = 3$ (dashed red).}
  \label{F3}
\end{figure}
When $\alpha_i = 0$, ${x_1}^* = \widetilde{x}_{1}$ and ${x_2}^* = 0$. Therefore, while ${x_1}^*$ is increasing in $\kappa_i$, ${x_2}^*$ is equal to zero for any $\kappa_i \in [0, 1]$. When $\alpha_i = 0.5$, the optimal strategy is characterized by two regions divided by $\kappa_i = 0.46$ (i.e., the value such that $\widetilde{x}_{1} = \underline{x_2}$). When $\kappa_i \in [0, 0.46)$, then ${x_1}^* = \widetilde{x}_{1}$ and ${x_2}^* = \underline{x_2}$, both increasing in $\kappa_i$. On the other hand, when $\kappa_i \in [0.46, 1]$, then ${x_1}^* = {x_2}^* = \hat{x}$, increasing in $\kappa_i$. Finally, when $\alpha_i = 3$, the optimal strategy is also characterized by two regions. When $\kappa_i \in [0, 0.03)$, then ${x_1}^* = x_s$ (constant in $\kappa_i$) and ${x_2}^* = \underline{x_2}$ (increasing in $\kappa_i$). When $\kappa_i \in [0.03, 1]$, then ${x_1}^* = {x_2}^* = \hat{x}$ (increasing in $\kappa_i$). Note that in $\kappa_i = 0.03$ (i.e., the value in which $u(x_s, \underline{x_2}) = u(\hat{x}, \hat{x})$), ${x_1}^*$ jumps up from $x_s$ to $\hat{x}$, while ${x_2}^*$ jumps down from $\underline{x_2}$ to $\hat{x}$. \par 
Several studies have examined the relationship between $x_1$ and $x_2$. Corollary \ref{C2} shows that if individuals attach a sufficiently large weight to universalization, then their offer is above their rejection threshold.
\begin{corol} \label{C2}
There exists $\tilde{\kappa}(\alpha_i) \geq 0$ such that for any $\kappa_i > \tilde{\kappa}(\alpha_i)$, ${x_1}^* \geq {x_2}^*$.      
\end{corol}
Corollary \ref{C2} can be interpreted as follows. When individuals attach a sufficiently high weight to universalization reasoning, their strategy does not have their offer below their rejection threshold. In other words, individuals avoid proposing offers they would reject in the \textit{responder} role. This finding aligns with \cite{ccelen2017blame}, who document that responders are more likely to reject offers lower than those they would have proposed in the \textit{proposer} role. Similarly, \cite{candelo2019proposer} finds a positive correlation between offers and hypothetical rejection thresholds, with the latter elicited in a post-experimental survey. \cite{ccelen2017blame} and \cite{candelo2019proposer} interpret these findings as evidence of reciprocity and the \textit{false consensus effect}, respectively. We argue that both observations can be explained through universalization reasoning.

\section{Types of punishers} \label{OUS}
In this section, we use data from \cite{van2023estimating} to better understand which individuals select high rejection thresholds in the UG. Although our theoretical model predicts individuals’ strategies, high rejection thresholds may arise from different combinations of parameters. In Section \ref{Pred}, we use estimates from \cite{van2023estimating} to compute predicted transfers in the DG and rejection thresholds in the UG at the individual level.\footnote{Besides the alternative sample restriction detailed in footnote (\ref{restriction}), the main difference from the estimation conducted in \cite{van2023estimating} is that we perform out-of-sample predictions for dictator and ultimatum games.} In Section \ref{predict}, we estimate a finite mixture model to characterize the behavior of different types. In our preferred specification, we identify two distinct types of individuals who select positive rejection thresholds: those who transfer approximately half of the endowment in the DG and those who transfer nothing.

We emphasize that we consider this exercise as qualitative rather than quantitative. This is motivated by two primary considerations. First, our theoretical model and the choice data from \cite{van2023estimating} differ in several dimensions (discussed below). This requires us to adopt further assumptions to make them directly comparable. Second, prior research highlights the challenges of making accurate quantitative out-of-sample predictions in contexts that differ from those used to estimate individual parameters (\citealp{erev1998predicting}).

\subsection{Predicted transfers and rejection thresholds} \label{Pred}

Before proceeding, a brief explanation of \cite{van2023estimating} is in order. 
The experiment was conducted at the CentERlab of Tilburg University with 136 subjects, primarily students.\footnote{The data from \cite{van2023estimating} used in this paper is available \href{https://dataverse.harvard.edu/dataset.xhtml?persistentId=doi:10.7910/DVN/PIQ572&version=1.0}{here}.} The experiment elicited subjects' behavior and beliefs across six sequential prisoner's dilemmas, six trust games, and six mini-ultimatum games. In each game, subjects made decisions for both roles behind the veil of ignorance. Additionally, subjects reported their beliefs about others' choices in each game. Each game involved two possible actions, with monetary payoffs varying according to the combination of players' choices.

In the mini-ultimatum games, the proposer chooses between an egalitarian distribution $(50,50)$ and an unequal split favoring themselves. If the proposer selects the unequal split, the responder decides whether to punish the proposer by opting for the egalitarian outcome $(10,10)$ or to accept the proposed unequal distribution. Across six versions of the game, the proposer’s unequal offers yielded $(60,40)$, $(65,35)$, $(70,30)$, $(75,25)$, $(80,20)$, and $(85,15)$, respectively. This design allows to compute each subject’s \textit{implicit rejection threshold}: the highest payoff in an unequal split that the subject still prefers to reject in favor of $(10,10)$ (more details below).

The authors use this data to estimate equation (\ref{1}) at various aggregation levels. For each level, they provide distinct sets of estimates based on differing assumptions regarding subjects' risk preferences.\footnote{Their individual-level estimates assume either risk-neutral or CRRA preferences, while their aggregate-level models assume either risk-neutrality or logarithmic utility.} They estimate their parameters at the experimental currency level (1 point $=$ $\EUR{0.17}$). Accordingly, we set $w = 58.8$ to represent an endowment of $\EUR{10}$.

Given Assumption \ref{concave} imposes risk aversion, we focus on the subjects from \cite{van2023estimating}'s \textit{core sample} whose individual CRRA utility function estimates satisfy $\alpha_i, \beta_i \in [-2, 2]$ and $\kappa_i \in [0, 1]$.\footnote{\cite{van2023estimating} define their \textit{core sample} as subjects whose individual risk-neutrality utility estimates satisfy $\alpha_i$, $\beta_i$, $\kappa_i$ $\in [-2, 2]$, resulting in a sample of 112 subjects. We introduce the additional restriction $\kappa_i \in [0, 1]$ and apply the same criteria under risk aversion. \label{restriction}} This results in a sample of 96 subjects, referred to as our \textit{test sample}.

Table \ref{tab:desc_stats_abcr} presents descriptive statistics for the individual estimates in our test sample (see Figures \ref{fig:alpha_hist_crra} to \ref{fig:kappa_hist_crra} for their distributions). On average, subjects in this sample are (i) spiteful when behind (mean $\alpha_i$ is 0.15), (ii) altruistic when ahead (mean $\beta_i$ is -0.16), and attach a positive weight towards universalization reasoning (mean $\kappa_i$ is 0.25). Substantial heterogeneity is observed in the estimated individual preference parameters. \vspace{0.2cm}
\begin{table}[H]
\centering
\begin{tabular}{lcccc}
\toprule
Variable & Min & Max & Mean & Median \\
\midrule
$\alpha$ & -0.87 & 1.07 & 0.15 & 0.15 \\
$\beta$ & -1.97 & 1.12 & -0.16 & -0.05 \\
$\kappa$ & 0.00 & 0.89 & 0.25 & 0.19 \\
\bottomrule
\end{tabular}
\caption{Descriptive statistics of individual estimates in the test sample ($N = 96$).}
\label{tab:desc_stats_abcr}
\end{table}

To compute individuals' predicted transfer and rejection threshold, we assume that they exhibit the preferences described in expression (\ref{1}), with $v(x) = ln(x+1)$.\footnote{Alternative specifications of $v(x)$ are problematic for several reasons. First, using $v(x) = x$ is not convenient, as transfers in the dictator game would be restricted to either $0$ or $\frac{w}{2}$. This prevents us from analyzing rejection thresholds for individuals who transfer intermediate amounts. Second, $v(x) = \frac{x^{1-\rho}}{1 - \rho}$ is also problematic, as it implies $v(0) < 0$. Although \cite{van2023estimating} use this specification, their games involve strictly positive monetary payoffs. \label{Explanation}

} The optimal transfer of individual $i$ in the DG is given by:
\begin{align}
x^*=&\argmax_{x \in [0, w]} \frac{1}{2}(1 - \kappa_i)\cdot v(w -  x) -\frac{1}{2} \alpha_i  \max \{v(x)-v(w-x),0\}  \\&- \frac{1}{2} \beta_i \max \{v(w-x)-v(x),0\} +\frac{1}{2} \kappa_i [v(w-x)+v(x)]. \nonumber
\end{align}

Figure \ref{fig:DG_histogram} shows the distribution of predicted transfers in the DG (see Table \ref{tab:descriptive_stats_donation_DG_x2bar} for the descriptive statistics). The average predicted transfer is approximately 16 \% of the endowment, lower than the 28.35\% in \cite{engel2011dictator} and the 40\% in \cite{branas2014fair}. Although the predicted transfers are lower than those reported by \cite{engel2011dictator}, the distribution is qualitatively similar (see Figure \ref{fig:DG_engel}). 

\begin{figure}[H]
    \centering
    \includegraphics[width=0.70\linewidth]{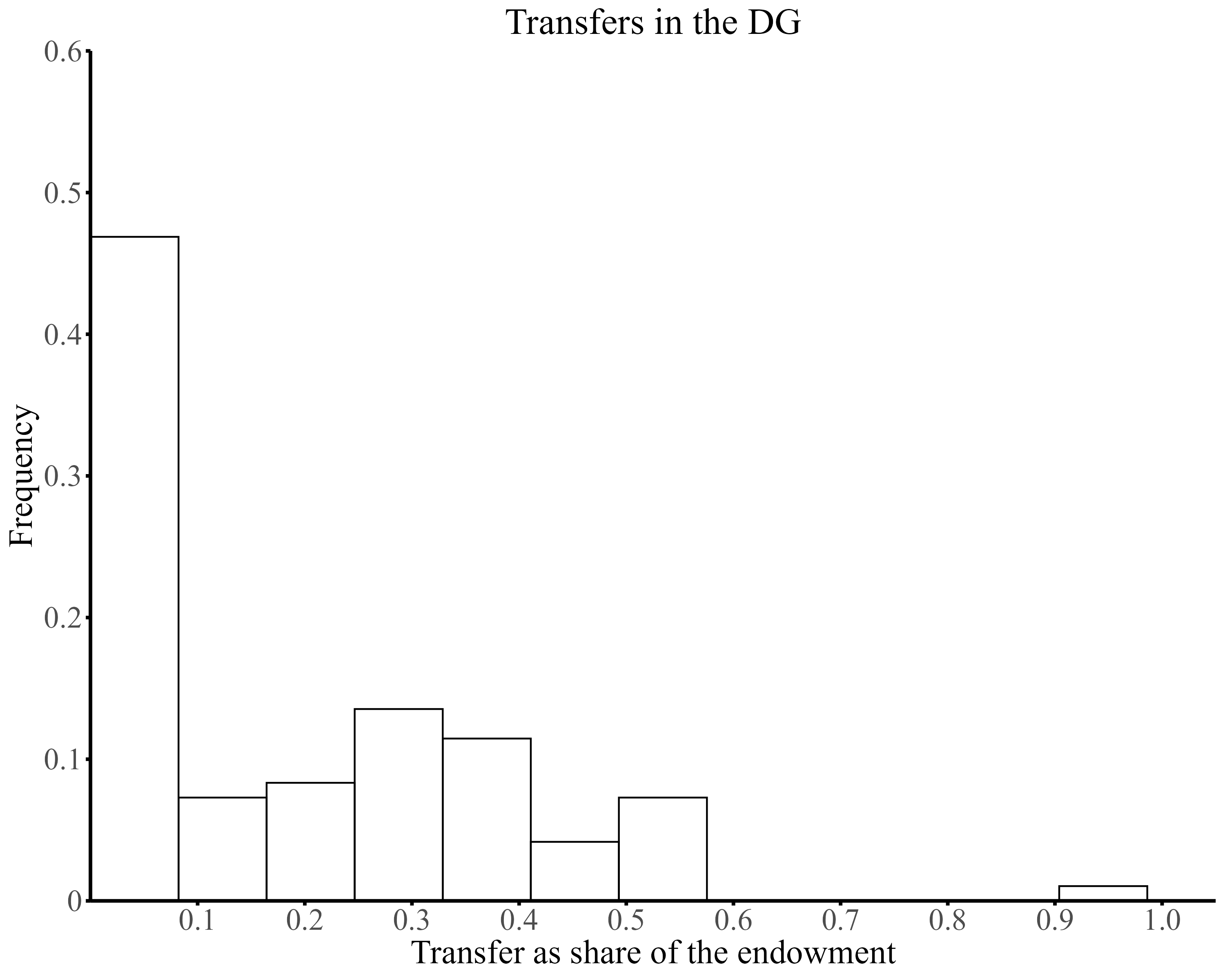}
    \caption{Predicted share transferred in the Dictator Game. Using \cite{van2023estimating} data, CRRA estimates, for $\alpha, \beta  \in [-2,2]$ and $\kappa \in [0,1]$. $N = 96$.}
    \label{fig:DG_histogram}
\end{figure}

To compute individuals' strategies in the UG, we assume that all individuals select ${x_1}^* = \frac{w}{2}$. Therefore, individuals' rejection thresholds are given by ${x_2}^* = \underline{x_2}$, as described in Definition \ref{D3}. This assumption is justified for three main reasons. First, there is a large concentration of offers in the UG around $\frac{w}{2}$.\footnote{For example, \cite{cochard2021} document a mean offer of 41.5\% in a meta-analysis with 96 studies of the UG.} Importantly, ${x_1}^* = \frac{w}{2}$ may be less influenced by individual preferences, as it can also reflect risk preferences, beliefs about others' rejection thresholds, or one's own hypothetical rejection threshold (see \citealp{candelo2019proposer}). Second, setting ${x_2}^* = \underline{x_2}$ is convenient as $\underline{x_2}$ does not depend on beliefs on others' offers. Third, as shown before, most individuals select a strategy with ${x_2}^* = \underline{x_2}$ under reasonable specifications (see Figure \ref{fig:dots}). Figure \ref{fig:histUG_w} presents the distribution of predicted rejection thresholds (see Table \ref{tab:descriptive_stats_donation_DG_x2bar} for the descriptive statistics).
\begin{figure}[H]
    \centering
\includegraphics[width=0.7\linewidth]{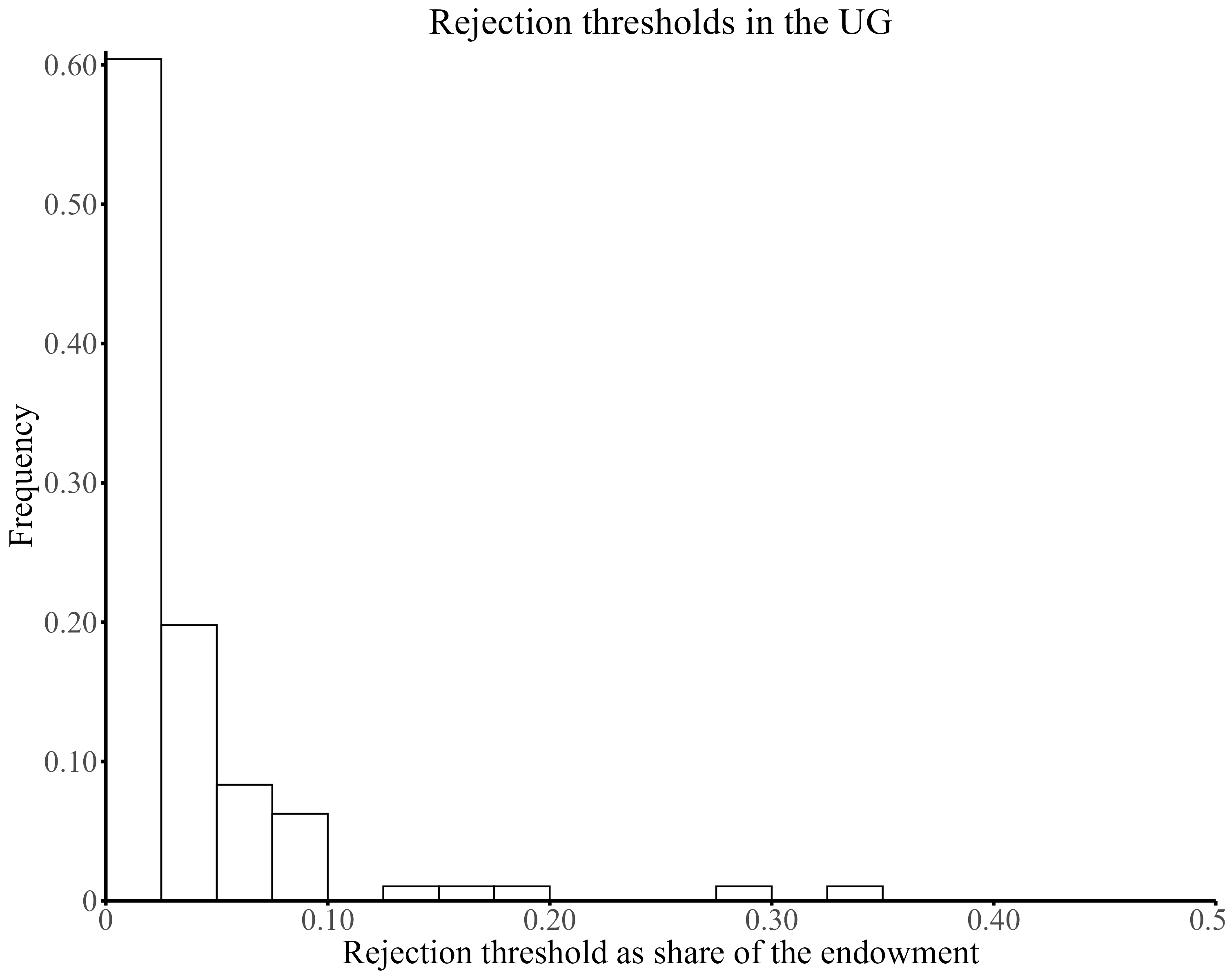}
    \caption{Constrained optimal rejection threshold $\underline{x_2}$ in the UG when $w = 58.8$ (10 euros). Using \cite{van2023estimating} core sample, CRRA estimates, for $\alpha, \beta  \in [-2,2]$ and $\kappa \in [0,1]$. $N = 96$.}
    \label{fig:histUG_w}
\end{figure}
The average rejection threshold is approximately 3\% of the endowment, substantially lower than the 33.35\% reported in \cite{branas2014fair}. This low average cannot be attributed to the inclusion of universalization reasoning given that rejection thresholds are increasing in $\kappa_i$ (see Figure \ref{fig:scatterUG_w}).\footnote{Table \ref{tab:descriptive_stats_donation_DG_x2bar_ab} shows that a model incorporating social preferences but excluding universalization reasoning predicts an average rejection threshold of  1\% of the endowment. } Despite the low average rejection threshold, more than half of the individuals are willing to reject positive amounts. 

We attribute the low rejection thresholds to three main reasons. First, approximately one-third of subjects in our sample exhibit $\alpha < 0$. For these individuals, accepting offers as low as zero is optimal (see Proposition \ref{1} and Figure \ref{fig:scatterUG_w}). Second, it is well-documented that people behave more pro-socially when deciding under role uncertainty (\citealp{iriberri2011role}; \citealp{alger2024doing}), which is the case in \cite{van2023estimating}. Finally, our estimates are computed from games with binary choices. This may reduce the likelihood of rejections, thereby generating parameter estimates consistent with lower rejection thresholds. Indeed, the \textit{implicit rejection thresholds} observed in \cite{van2023estimating} are relatively low, with 81\% of subjects accepting splits as unequal as 85-15, and an overall acceptance rate of 91\% (see Table 1 in \citealp{van2023estimating}). Thus, our predicted behavior is more in line with that observed in \cite{van2023estimating} than in \cite{branas2014fair}.\footnote{Figure \ref{fig:scatter_plot_w} illustrates the relationship between predicted transfers in the DG and rejection thresholds in the UG. For comparison, actual transfers and thresholds reported in \cite{branas2014fair} are presented in Figures \ref{fig:scatter_plot_branas_lab} and \ref{fig:scatter_plot_branas_city}. Figures \ref{fig:histDG_w_ab_all} to \ref{fig:scatter_plot_w_ab_all} present predicted transfers and rejection thresholds assuming $\kappa = 0$. While predicted transfers in the DG are similar, rejection thresholds display substantially less variability and are more concentrated near zero. Finally, Figure \ref{fig:scatter_plot_w_c_all} highlights the crucial role of social preferences in explaining differences in rejection thresholds in the UG.}

\subsection{3-Type predicted behavior} \label{predict}
Previous literature identifies two types of punishers in the UG: \textit{prosocial} punishers, who are willing to sacrifice their own material payoff to enforce fairness norms, and \textit{spiteful} punishers, who assign a negative weight to others’ material payoffs. For example, \cite{branas2014fair} document that individuals who transfer either zero or half of the endowment are most likely to select high rejection thresholds. 

In this section, we use the choice data from \cite{van2023estimating} to structurally estimate the utility parameters ($\alpha$, $\beta$, $\kappa$) and the ``noise'' parameter $\lambda$ of our test sample ($N = 96$).\footnote{This noise parameter $\lambda$ reflects the probability that decisions are made according to the postulated utility function, with a lower $\lambda$ indicating that choices follow closely the posited function.} We use finite mixture models (following \citealp{bruhin2019} and \citealp{van2023estimating}). These models assume that there exists a finite number of types $K \geq 1$ in the population and provide estimates of  $\alpha$, $\beta$, $\kappa$ and $\lambda$ for each type, along with its population share. We estimate a representative agent model ($K = 1$) and the cases with $K = 2$ and $K = 3$ using the expectation-maximization algorithm (\citealp{mclachlan2019}).\footnote{For a detailed discussion of the method, see \cite{van2023estimating}.}

The upper part of Table \ref{tab:paujose_core} reports the model results across different levels of aggregation. Our estimates are remarkably similar to those reported in \cite{van2023estimating} when imposing logarithmic utility (see their Table 5). When assuming a representative agent ($K = 1$), the estimates for the parameters are $\alpha_0 = 0.14$, $\beta_0 = -0.01$, and $\kappa_0 = 0.22$, with the subscript $0$ indicating the representative agent. These estimates are consistent with Assumption \ref{Types}.

When we allow for multiple types, all types exhibit positive degrees of morality ($\kappa > 0$), emphasizing the importance of this concern for understanding individuals’ behavior. In the two-type model, Type 1 (61\% of the sample) exhibits weak spite both when behind and ahead ($\alpha_1 = 0.05$, $\beta_1 = 0.08$) and strong universalization concerns ($\kappa_1 = 0.25$). Type 2 displays strong spite both when behind and ahead ($\alpha_2 = 0.28$, $\beta_2 = 0.28$) and lower universalization concerns ($\kappa_2 = 0.16$). In the three-type model, Type 1 (29\% of the sample) exhibits spite both when behind and ahead ($\alpha_1 = 0.13$, $\beta_1 = 0.22$) and strong universalization concerns ($\kappa_1 = 0.26$). Type 2 (33\% of the sample) exhibits no social preferences, with both coefficients not significantly different from zero, and moderate universalization concerns ($\kappa_2 = 0.22$). Finally, Type 3 is essentially identical to Type 2 in the two-type model and accounts for a similarly sized fraction of subjects (38\% versus 39\% in the two-type model). 

Based on the aggregate preference estimates, we compute the predicted transfers in the DG and the rejection thresholds in the UG. The last two rows of Table \ref{tab:paujose_core} display these predictions for each type. Given that the three-type model is the preferred specification in \cite{van2023estimating}, we focus our analysis on its results. Moreover, when comparing model fit using the \textit{Integrated Completed Likelihood} (ICL) and the \textit{Normalized Entropy Criterion} (NEC), our findings align with those of \cite{van2023estimating}: the NEC favors the two-type model, while the ICL favors the three-type model.\footnote{For both metrics, lower values indicate a better model fit. See \cite{celeux1996} and \cite{biernacki2000} for detailed derivations of the NEC and ICL metrics, respectively, and \cite{mclachlan2019} for a discussion on using these criteria to evaluate finite mixture models..}

\begin{table}[H]
    \begin{center}
    \caption{Aggregate estimates of preference ($\alpha$, $\beta$, $\kappa$) and noise ($\lambda$) parameters, type shares and predicted transfers and rejection thresholds for \cite{van2023estimating} core sample, with individual estimates of $\alpha, \beta \in [-2, 2]$ and $\kappa \in [0, 1]$. $N = 96$.}
        \label{tab:paujose_core}
    \begin{threeparttable}
    \begin{adjustbox}{width=\textwidth}
        \begin{tabular}{lcccccc}
            \toprule
            & \multicolumn{1}{c}{One type} & \multicolumn{2}{c}{Two types} & \multicolumn{3}{c}{Three types} \\
            \cmidrule(lr){2-2} \cmidrule(lr){3-4} \cmidrule(lr){5-7}
            & Type 1 & Type 1 & Type 2 & Type 1 & Type 2 & Type 3 \\
            \midrule
            $\alpha$                                & 0.14 (0.03) & 0.05 (0.03) & 0.28 (0.04) & 0.13 (0.08) & -0.02 (0.06) & 0.28 (0.03) \\
            $\beta$                                 & -0.01 (0.03) & 0.08 (0.03) & -0.30 (0.08) & 0.22 (0.06) & -0.08 (0.07) & -0.30 (0.09) \\
            $\kappa$                                & 0.22 (0.01) & 0.25 (0.02) & 0.19 (0.02) & 0.26 (0.05) & 0.22 (0.04) & 0.19 (0.02) \\
            $\lambda$                               & 0.25 (0.01) & 0.28 (0.02) & 0.16 (0.02) & 0.21 (0.02) & 0.33 (0.05) & 0.16 (0.02) \\
            Type share            & 1.00 & 0.61 (0.05) & 0.39 (0.05) & 0.29 (0.07) & 0.33 (0.08) & 0.38 (0.06) \\
            \midrule
            \textbf{Model fit} \\
            $ln(\mathcal{L})$                                     & \multicolumn{1}{c}{-2,063.28} & \multicolumn{2}{c}{-1,902.76} & \multicolumn{3}{c}{-1,865.90} \\
            EN($\tau$)                                   & \multicolumn{1}{c}{0.00} & \multicolumn{2}{c}{4.00} & \multicolumn{3}{c}{13.50} \\
            ICL                                     & \multicolumn{1}{c}{4,144.82} & \multicolumn{2}{c}{3,850.59} & \multicolumn{3}{c}{3,809.21} \\
            NEC                                     & \multicolumn{1}{c}{---} & \multicolumn{2}{c}{0.02} & \multicolumn{3}{c}{0.07} \\
            \midrule
            \textbf{Predicted behavior} \\
            DG transfer                             & 7.80 (0.13) & 14.90 (0.25) & 0.00 (0.00) & 22.10 (0.38) & 6.17 (0.10) & 0.00 (0.00) \\
            UG rejection threshold                  & 0.84 (0.01) & 0.31 (0.01) & 1.85 (0.03) & 0.85 (0.01) & 0.00 (0.00) & 1.81 (0.03) \\
            \bottomrule
        \end{tabular}
            \end{adjustbox}
        \end{threeparttable}
    \end{center}
        \footnotesize\textit{Notes:}  The table reports parameter estimates for one-, two-, and three-type models. Bootstrapped standard errors (SE) are displayed in parentheses next to the estimates. The DG transfer and UG rejection threshold values are shown as shares of the total endowment in parentheses. $ln(\mathcal{L})$: Log-likelihood; EN($\tau$): Entropy; $ICL$: Integrated completed likelihood; $NEC$: Normalized entropy criterion
\end{table}

In the three-type model, two types of individuals select positive rejection thresholds: one that transfers nearly 40\% of the endowment (Type 1) and another that transfers zero (Type 3). These two types differ primarily in their social preferences. Type 3 exhibits twice the level of spitefulness as Type 1. In contrast, Type 2, characterized solely by universalization reasoning, transfers 10\% of the endowment and accepts any offer in the UG This is in line with the behavior documented in \cite{branas2014fair}.

The predicted behavior reported in Table \ref{tab:paujose_core} aligns with the key insights from Section \ref{sec:theory}: (i) rejections in the UG are driven by spite, and (ii) universalization reasoning amplifies social preferences. Specifically, individuals with small but positive $\alpha_i$ and relatively large $\kappa_i$ (e.g., Type 1 in the two-type model) exhibit high transfers and positive rejection thresholds. In contrast, individuals with high and positive $\alpha_i$ and $\beta_i$ but moderate $\kappa_i$ (e.g., Type 2 in the two-type model or Type 3 in the three-type model) display no transfers in the DG but maintain positive rejection thresholds in the UG.

\section{Conclusions} \label{S5}
In this paper, we study the motivations behind individuals' rejections in the UG. We characterize the optimal offers and rejection thresholds of individuals whose utility function combine social preferences and moral concerns. We derive three main theoretical results. First, spite plays a crucial role in the rejections observed. Specifically, under the preferences considered, a positive degree of spite is both a necessary and sufficient condition for selecting a positive rejection threshold. In contrast, a positive degree of universalization reasoning is not even necessary for doing so. Second, both offers and rejection thresholds are increasing with individuals' universalization reasoning. Intuitively, universalization reasoning amplifies social preferences and is therefore only active when such preferences are present. Third, when individuals assign a sufficiently large weight to universalization reasoning, their offers exceed their rejection thresholds.

These three predictions are consistent with existing empirical evidence. They help explain (i) the rejection of positive offers by individuals exhibiting either selfish or prosocial behavior in the DG, (ii) the correlation between rejection thresholds and measures of morality (\citealp{kimbrough2016norms}; \citealp{capraro2022moral}), (iii) why individuals are more likely to reject offers if those are lower than the offer they would have selected as proposers (\citealp{ccelen2017blame}), and (iv) the positive correlation between offers and (hypothetical) rejection thresholds (\citealp{candelo2019proposer}).

To better understand the motives behind individuals’ rejections, we use the choice data from \cite{van2023estimating} to estimate a finite mixture model that characterizes different behavioral types. This approach allows us to predict each type’s transfer in the DG and rejection threshold in the UG. Focusing on the three-type model, the data clearly identifies two types who reject offers in the UG and differ in their predicted behavior in the DG. The first type is predicted to transfer nearly 40\% of the endowment, while the second type is predicted to transfer nothing. These two types differ primarily in their degree of spite: although both exhibit spiteful concerns, the second type is approximately twice as spiteful as the first. In contrast, the third type, labeled the “non-punisher”, lacks social preferences but exhibits universalization reasoning.

\newpage
\begin{appendices}
\section{Mathematical proofs} \label{AppA}
\subsection{Proofs main text}
When individual $i$ selects strategy $x = (x_1, x_2)$, she gets an expected utility of 
\begin{eqnarray} \label{eq: (6)} 
{u}_{i}(x_1, x_2)&=&
(1 - \kappa_i)v(w - x_1)F^{i}_{y_2}(x_1) \\
&+& \int_{x_2}^{\frac{w}{2}} [(1 - \kappa_i)v(y_1) - \alpha_i (v(w - y_1) - v(y_1))]dF_{i}^{y_1} \nonumber \\
&+&  \kappa_i \mathbf{1}_{\{x_1 \geq x_2\}}[v(w - x_1) + v(x_1)],
\nonumber 
\end{eqnarray} \par 
In the main text, we defined the four following objects. 
\begin{itemize}
    \item \textbf{Definition 1} (\emph{Selfish Offer}): $x_{s} \equiv \argmax_{x \in [0, w]}  \frac{1}{2}v(w - x_1)F^{i}_{y_2}(x_1)$. 

    \item \textbf{Definition 2} (\emph{Constrained optimal offer}): $\tilde{x}_{1}(\kappa_i) \equiv \argmax_{x_1 \in [0, w]} \frac{1}{2}(1 - \kappa_i) v(w - x_1)F^{i}_{y_2}(x_1) + \frac{1}{2}\kappa_i(v(w - x_1) + v(x_1))$.

    \item \textbf{Definition 3} (\emph{Constrained optimal rejection threshold}): $\underline{x_2}(\kappa_i, \alpha_i) \equiv  \argmax_{x_2 \in [0, w]} \int_{x_2}^{\frac{w}{2}} [(1 - \kappa_i)v(y_1) - \alpha_i (v(w - y_1) - v(y_1))]dF_{i}^{y_1}$.

    \item \textbf{Definition 4} (\emph{Symmetric optimal strategy}): $\hat{x}(\kappa_i, \alpha_i) \equiv \argmax_{x \in [0, w]}u_{i}(x, x)$.
\end{itemize} \par 
As in the main text, we omit the parameters from the above-defined objects and refer to them as $x_{s}$, $\tilde{x}_{1}$, $\underline{x_2}$, and $\hat{x}$. \par 
In Assumption \ref{A4}, we assumed that $u_{i}(x_1, x_2)$ is strictly concave in $x_1$ (i.e., increasing in $x_1$ over the interval $[0, \tilde{x}_{1})$ and decreasing in $x_1$ over the interval $(\tilde{x}_{1}, w]$). We now show under which conditions this is the case. To do so, it is sufficient to show that when $x_1 \geq x_2$, $\frac{\partial^2 {u}_{i}(x_1, x_2)}{\partial^2 x_1} < 0$.
\begin{eqnarray}
   \frac{\partial {u}_{i}(x_1, x_2)}{\partial x_1} &=& -\frac{1}{2} (1 - \kappa_i) v'(w - x_1)F^{i}_{y_2}(x_1) + \frac{1}{2} (1 - \kappa_i) v(w - x_1)f_{y_2}^{i}(x_1) \\
   &+& \frac{1}{2}\kappa_i \left[v'(x_1) - v'(w - x_1)\right], \nonumber
\end{eqnarray}
\begin{eqnarray} \label{eq: (7)}
\frac{\partial^2 {u}_{i}(x_1, x_2)}{\partial^2 x_1} & = &
\frac{1}{2} (1 - \kappa_i) v''(w - x_1)F_{y_2}^{i}(x_1) - \frac{1}{2} (1 - \kappa_i) v'(w - x_1)f_{y_2}^{i}(x_1)  \\
&+&  \frac{1}{2} (1 - \kappa_i) v(w - x_1)\frac{\partial f_{y_2}^{i}(x_1)}{\partial x_1} +  \frac{1}{2}\kappa_i \left[v''(w - x_1) + v''(x_1)\right]. \nonumber 
\end{eqnarray}
All terms in (\ref{eq: (7)}), except for the third, are strictly negative. Therefore, $\frac{\partial^2 {u}_{i}(x_1, x_2)}{\partial^2 x_1} < 0$ as long as $\frac{\partial f_{y_2}^{i}(x_1)}{\partial x_1}$ is below a certain positive constant, as discussed in footnote \ref{Convavity}.
\par 
\begin{proof}
\textit{Lemma 1}: Consider two cases: i) $\tilde{x}_{1} \geq \underline{x_2}$ and ii) $ \tilde{x}_{1} < \underline{x_2}$. \par 
For the first case, ${x_1}^* = \tilde{x}_{1}$ and ${x_2}^* = \underline{x_2}$. This follows from the definitions of $\tilde{x}_{1}$ and $\underline{x_2}$, as conditional on having $\tilde{x}_{1} \geq \underline{x_2}$, $\tilde{x}_{1}$ and $\underline{x_2}$ are the offer and rejection threshold that maximizes individual $i$'s utility. \par 
For the second case, we first show the three following claims. In all claims, we fix $\alpha_i$ and $\kappa_i \in [0, 1]$ such that $\underline{x_2} > \tilde{x}_{1}$. \par 
\textbf{Claim 1}: ${x_1}^* \leq {x_2}^*$. \par  
By contradiction, suppose that there exists a strategy $x^*$ with ${x_1}^* > {x_2}^*$. By the definition of $\underline{x_2}$, ${x_2}^* = \underline{x_2}$. Additionally, by the strictly concavity of $u_{i}(x_1, x_2)$ on $x_1$, when ${x_1} > {x_2}$ ${u}_{i}(x_1, x_2)$ is strictly decreasing in $x_1$ for any $x_1 > \tilde{x}_{1}$. But then, when ${x_2}^* = \underline{x_2} > \tilde{x}_{1}$, choosing ${x_1}^* = {x_2}^*$ gives a strictly larger utility than choosing any ${x_1}^* > {x_2}^*$.  \par 
\textbf{Claim 2}: If ${x_1}^* < {x_2}^*$, then ${x_1}^* = x_s$ and ${x_2}^* = \underline{x_2}$. \par
Note that all strategies $x$ that satisfy ${x_1} < {x_2}$ have the same value in the third term in (\ref{eq: (6)}) (as $ \mathbf{1}_{\{x_1 \geq x_2\}} = 0$). Therefore, conditional on choosing a strategy $x$ with ${x_1} < {x_2}$, the individual selects the offer and rejection threshold that maximizes his expected payoff (i.e., ${x_1}^* = x_s$ and ${x_2}^* = \underline{x_2})$. \par 
\textbf{Claim 3}: If ${x_1}^* = {x_2}^*$, then ${x_1}^* = {x_2}^* \in [\tilde{x}_{1}, \underline{x_2}$]. \par 
We need to show that there cannot exist a strategy $x^*$ with ${x_1}^* = {x_2}^* < \tilde{x}_{1}$ or ${x_1}^* = {x_2}^* >\underline{x_2}$. We rule out the first case by contradiction. But if ${x_1}^* = {x_2}^* < \tilde{x}_{1}$, then $x = (\tilde{x}_{1}, {x_2}^*)$ gives a strictly higher utility to the individual as $x_1 = \tilde{x}_{1}$ maximizes his utility when $x_1 \geq x_2$. We also rule out the second case by contradiction. But if ${x_1}^* = {x_2}^* >\underline{x_2}$, then $x = ({x_1}^*, \underline{x_2})$ gives a strictly larger utility to the individual as $x_2 = \underline{x_2}$ maximizes his utility when $x_1 \geq x_2$.  \par 
With the three previous claims, we can conclude that to compute $x^*$, when $\underline{x_2} > \tilde{x}_{1}$, we need to compare i) $u(x_s, \underline{x_2})$ and ii) $u(\hat{x}, \hat{x}) =\argmax_{y \in [\tilde{x}_{1}, \underline{x_2}]} u(y, y)$. When \\ $u(x_s, \underline{x_2}) > u(\hat{x}, \hat{x}$), then ${x_1}^* = x_s$ and ${x_2}^* = \underline{x_2}$, while ${x_1}^* = {x_2}^* = \hat{x}$ otherwise.  
\end{proof}
\begin{proof}
\textit{Lemma 2}: We first show that there exists $\tilde{\kappa}(\alpha_i)$ such that: (i) when $\kappa_i < \tilde{\kappa}(\alpha_i)$, $u(x_s, \underline{x_2}) > u(\hat{x}, \hat{x})$, (ii) when $\kappa_i = \tilde{\kappa}(\alpha_i)$, $u(x_s, \underline{x_2}) = u(\hat{x}, \hat{x})$, and (iii) when $\kappa_i > \tilde{\kappa}(\alpha_i)$, $u(x_s, \underline{x_2}) < u(\hat{x}, \hat{x})$. \par
To show the result, three remarks are important to emphasize. First,
\begin{eqnarray}
u(x_s, \underline{x_2}) &=& \frac{1}{2} (1 - \kappa_i) v(w - x_s)F^{i}_{y_2}(x_1) \\ &+& \frac{1}{2}\int_{\underline{x_2}}^{\frac{w}{2}} [(1 - \kappa_i + \alpha_i)v(y_1) - \alpha_i v(w - y_1)]f^{i}(y_1) \,dy_1  \nonumber  
\end{eqnarray}
is strictly decreasing in $\kappa_i$. By the Envelope theorem,
\begin{equation}
  \frac{\partial u(x_s, \underline{x_2})}{\partial \kappa_i} = -\frac{1}{2} v(w - x_s)F^{i}_{y_2}(x_1) -\frac{1}{2} \int_{\underline{x_2}}^{\frac{w}{2}} v(y_1) f^{i}(y_1) \, dy_1 < 0 
\end{equation}
Second, when $\kappa_i = 0$, $u(x_s, \underline{x_2}) \geq u(\hat{x}, \hat{x})$. Third, if $x_s < \frac{w}{2}$, $u(x_s, \underline{x_2}) < u(\hat{x}, \hat{x})$ when $\kappa_i = 1$.\footnote{When $x_s = \frac{w}{2}$, then $x_1 = \frac{w}{2}$ and $x_2 = \underline{x_2}$ is the optimal strategy for any $\kappa_i$ and $\alpha_i$.} Then, we need to determine how $u(\hat{x}, \hat{x})$ depends on $\kappa_i$. Recall that $u(\hat{x}, \hat{x})$ is defined as follows: 
\begin{eqnarray}
u(\hat{x}, \hat{x}) & = &
\frac{1}{2}(1 - \kappa_i)v(w -  \hat{x})F^{i}_{y_2}(\hat{x}) \\
&+&  \frac{1}{2}\int_{\hat{x}}^{\frac{w}{2}}[(1 - \kappa_i + \alpha_i)v(y_1) - \alpha_i v(w - y_1)]f^{i}(y_1) \,dy_1 \nonumber \\
&+&  \frac{1}{2}\kappa_i \left[v(w - \hat{x}) + v(\hat{x})\right],    \nonumber
\end{eqnarray}
where $\hat{x} = \argmax_{y \in [\tilde{x}_{1}, \underline{x_2}]} u(y, y)$. By the Envelope Theorem, 
\begin{eqnarray}
\frac{\partial u(\hat{x}, \hat{x})}{\partial \kappa_i} & = &
-\frac{1}{2}v(w - \hat{x})F^{i}_{y_2}(\hat{x}) \\ \nonumber
&-&  \frac{1}{2}\int_{\hat{x}}^{\frac{w}{2}} v(y_1)f^{i}(y_1) \,dy_1 \\ \nonumber 
&+& \frac{1}{2}\left[v(w - \hat{x}) + v(\hat{x})\right].  \nonumber 
\end{eqnarray}
To determine the sign of $\frac{\partial u(\hat{x}, \hat{x})}{\partial \kappa_i}$, we examine its derivative:
\begin{eqnarray}
\frac{\partial^2 u(\hat{x}, \hat{x})}{\partial^2 \kappa_i} & = &
\frac{1}{2}\frac{\partial \hat{x}}{\partial \kappa_i}\left[v'(w - \hat{x})F^{i}_{y_2}(\hat{x}) - v(w - \hat{x})f^{i}_{y_2}(\hat{x})\right] \\ \nonumber
&+&  \frac{1}{2}\frac{\partial \hat{x}}{\partial \kappa_i}v(\hat{x})f^{i}(\hat{x}) \\ \nonumber 
&+&  \frac{1}{2}\frac{\partial \hat{x}}{\partial \kappa_i}\left[v'(\hat{x}) - v'(w - \hat{x})\right]. \nonumber 
\end{eqnarray}
The sign of $\frac{\partial^2 u(\hat{x}, \hat{x})}{\partial^2 \kappa_i}$ depends on the sign of $\frac{\partial \hat{x}}{\partial \kappa_i}$. We know that $\hat{x}$ is computed from
\begin{eqnarray}
\max_{y \in [\tilde{x}_{1}, \underline{x_2}(\gamma_i,\alpha_i)]} &&
\frac{1}{2}(1 - \kappa_i)v(w -  y)F^{i}_{y_2}(y) \\
&+&  \frac{1}{2}\int_{y}^{\frac{w}{2}}[(1 - \kappa_i + \alpha_i)v(y_1) - \alpha_i v(w - y_1)]f^{i}(y_1) \,dy_1 \nonumber \\
&+&  \frac{1}{2}\kappa_i \left[v(w - y) + v(y)\right].    \nonumber
\end{eqnarray}
The first-order condition of the previous expression is given by
\begin{eqnarray} \label{11}
\frac{1}{2}(1 - \kappa_i)[-v'(w - \hat{x})F^{i}_{y_2}(\hat{x}) + v(w - \hat{x})f^{i}_{y_2}(\hat{x})]  \\
- \frac{1}{2}f^{i}_{y_1}(\hat{x}) [v(\hat{x})(1 - \kappa_i + \alpha_i) - \alpha_i v(w -\hat{x})]\nonumber \\
+ \frac{1}{2}\kappa_i \left[v'(\hat{x}) - v'(w - \hat{x})\right]. \nonumber
\end{eqnarray}
 
Assumption \ref{A4} guarantees that $\hat{x}$ is a maximum of $u(\hat{x}, \hat{x})$. Following previous derivations, we compute the implicit derivative (\ref{11}) to obtain that $\frac{\partial \hat{x}}{\partial \kappa_i} \geq 0$. This implies that $u(\hat{x}, \hat{x})$ is convex in $\kappa_i$. This, jointly with (i) $u(x_s, \underline{x_2}) \geq u(\hat{x}, \hat{x})$ when $\kappa_i = 0$ and (ii) $u(x_s, \underline{x_2}) < u(\hat{x}, \hat{x})$ when $\kappa_i = 1$ guarantees the existence of a unique $\tilde{\kappa}(\alpha_i) \geq 0$ such that when $\kappa_i < \tilde{\kappa}(\alpha_i)$, $u(x_s, \underline{x_2}) > u(\hat{x}, \hat{x})$, and $u(x_s, \underline{x_2}) < u(\hat{x}, \hat{x})$ otherwise. To see this, note that $u(\hat{x}, \hat{x})$ convex implies that there can only exist a unique global minimum of $u(\hat{x}, \hat{x})$ (e.g. at $x_1 = x_2 = \hat{x}(\tilde{\kappa})$ with $\tilde{\kappa} > 0$). Then, $u(\hat{x}, \hat{x})$ is decreasing for $\kappa_i \in [0, \tilde{\kappa})$ and increasing for $\kappa_i > \tilde{\kappa}$. Thus, $u(x_s, \underline{x_2}) \geq u(\hat{x}, \hat{x})$ when $\kappa_i = 0$ implies that $u(x_s, \underline{x_2}) > u(\hat{x}, \hat{x})$ when $\kappa_i = \tilde{\kappa}$, which shows the result.  \par  
We now show that there exists a function $\tilde{\alpha}(\kappa_i) \equiv  \frac{(1 - \kappa_i)v(\tilde{x}(\kappa_i))}{v(w - \tilde{x}(\kappa_i))- v(\tilde{x}(\kappa_i))} > 0$ such that: (i) when $\alpha_i = \tilde{\alpha}(\kappa_i)$, $\tilde{x}_{1}(\kappa_i) = \underline{x_2}$, (ii) when $\alpha_i > \tilde{\alpha}(\kappa_i)$, $\tilde{x}_{1}(\kappa_i) < \underline{x_2}$, and (iii) when $\alpha_i < \tilde{\alpha}(\kappa_i)$, $\tilde{x}_{1}(\kappa_i) > \underline{x_2}$. Note that $\tilde{x}_{1}(\kappa_i) = \underline{x_2}$ is equivalent to $v(\tilde{x}_{1})(1 + \alpha_i - \kappa_i) - \alpha_i v(w - \tilde{x}_{1}) = 0$, or equivalently $\alpha_i = \frac{(1 - \kappa_i)v(\tilde{x}(\kappa_i))}{v(w - \tilde{x}(\kappa_i))- v(\tilde{x}(\kappa_i))}$. Fixing $\kappa_i$ such that $\tilde{x}_{1}(\kappa_i) = \underline{x_2}$, increasing $\alpha_i$, increases $\underline{x_2}$, but not $\tilde{x}_{1}(\kappa_i)$. Thus, for any $\alpha > \tilde{\alpha}(\kappa_i)$, $\tilde{x}_{1}(\kappa_i) < \underline{x_2}$.\footnote{Note that $\tilde{\alpha}(\kappa_i) > 0$ as for any $\kappa < 1$, $\tilde{x}_{1}(\kappa) \in [x_s, \frac{w}{2})$ (and therefore $v(w - \tilde{x}(\kappa))- v(\tilde{x}(\tilde{\kappa})) > 0$).}    
\end{proof}

\begin{proof}
\textit{Proposition \ref{P1}}: Follows from combining Lemmas 1 and 2.    
\end{proof}
\begin{proof}
\textit{Corollary \ref{C1}}: To show the first result, we fix $\alpha_i = 0$. Individual $i$ maximizes
\begin{eqnarray} \label{10}
    &&\max_{\{x_1, x_2\} \in [0, w]^2} \frac{1}{2}(1 - \kappa_i)v(w - x_1)F_{y_2}^{i}(x_1) + \frac{1}{2}\int_{x_2}^{\frac{w}{2}}(1 - \kappa_i)v(y_1)f^{i}(y_1)\,dy_1 \\ &+& \kappa_i \mathbf{1}_{\{x_1 \geq x_2\}}\frac{1}{2}\left[v(w - x_1) + v(x_1)\right]. \nonumber
\end{eqnarray} \par 
The first term in (\ref{10}) does not depend on $x_2$, the second term is maximized at $x_2 = 0$, and the third term is maximized at any $x_2 \in [0, x_1]$. Thus, when $\alpha_i = 0$ and $\kappa_i \in [0, 1)$, choosing $x_2 = 0$ maximizes individual's utility. When $\kappa_i = 1$, then ${x_1}^* = \frac{w}{2}$ and ${x_2}^* \in [0, \frac{w}{2}]$. Therefore, it is not true that when $\alpha = 0$ and $\kappa = 1$, the optimal strategy cannot have a rejection threshold above 0. \par
We now show that when $\alpha_i > 0$, ${x_2}^* > 0$ for any $\kappa_i \in [0, 1)$. By contradiction, suppose that an individual with $\alpha_i > 0$ and $\kappa_i \in [0, 1)$ has an optimal strategy $x^*$ with ${x_2}^* = 0$. By Definition 3, we know that $\underline{x_2}(\kappa_i, \alpha_i) > 0$ when $\alpha_i > 0$ and $\kappa_i \in [0, 1)$. Therefore, by Proposition \ref{P1}, $x^*$ must be symmetric (as otherwise ${x_2}^* > 0$). However, selecting ${x_1}^* = {x_2}^* = 0$ is not an optimal strategy as would be dominated by, e.g., ${x_1}^* = x_s$ and ${x_2}^* = 0$. \par
\noindent \textit{Corollary \ref{C3}}: We distinguish two cases depending if (i) $\alpha_i \in (\tilde{\alpha}(\kappa_i), \overline{\alpha})$ or if (ii) $\alpha_i \in (0, \tilde{\alpha}(\kappa_i))$. \par
By Proposition \ref{P1}, if $\alpha_i \in (\tilde{\alpha}(\kappa_i), \overline{\alpha})$, then $x^* = (\hat{x}, \hat{x})$. $\frac{\partial \hat{x}}{\partial \kappa_i} \geq 0$ follows from taking the implicit derivative of (\ref{11}). \par 
By Proposition \ref{P1}, if $\alpha_i \in (0, \tilde{\alpha}(\kappa_i))$, $\tilde{x}_{1} \geq \underline{x_2}$. We need to show that (i) $\frac{\partial \tilde{x}_{1}}{\partial \kappa_i} \geq 0$ and (ii) $\frac{\partial \underline{x_2}}{\partial \kappa_i} \geq 0$.
\begin{equation} \footnotesize
\frac{\partial \tilde{x}_{1}}{\partial \kappa_i} = \frac{-\kappa_i \left[v'(\tilde{x}_{1}) - v'(w - \tilde{x}_{1})\right]}{(1 - \kappa_i) \left[v''(w - \tilde{x}_{1})F_{y_2}^{i}(\tilde{x}_{1}) - v'(w - \tilde{x}_{1})f_{y_2}^{i}(\tilde{x}_{1}) + v(w - \tilde{x}_{1})\frac{\partial f_{y_2}^{i}(\tilde{x}_{1})}{\partial x_1}\right] +  \kappa_i\left[v''(w - \tilde{x}_{1}) + v''(\tilde{x}_{1})\right]} \geq 0,
\end{equation}
as both the numerator and denominator are negative.\footnote{The numerator is negative as $v'(\tilde{x}_{1}) - v'(w - \tilde{x}_{1}) > 0$, while the denominator is negative because of Assumption \ref{A4} (see equation (\ref{eq: (7)})).}
\begin{equation}
\frac{\partial \underline{x_2}}{\partial \kappa_i} = \frac{(1 - \kappa_i)\alpha_i\left[v(w - \underline{x_2}) - v(\underline{x_2})\right]}{(1 - \kappa_i)v'(\underline{x_2}) + \alpha_i \left[v'(\underline{x_2}) + v'(w - \underline{x_2})\right]} \geq 0,     
\end{equation}
as both the numerator and denominator are positive.\footnote{The numerator is positive as $v(w - \underline{x_2}) - v(\underline{x_2}) > 0$, while the denominator is positive as $v' > 0$.}
\end{proof}
\begin{proof}
\textit{Corollary \ref{C2}}: By Proposition \ref{P1}, ${x_2}^* > {x_1}^*$ when $\alpha_i > \overline{\alpha}$ and $\kappa_i \leq \tilde{\kappa}(\alpha_i)$. Therefore, for any $\kappa_i > \tilde{\kappa}(\alpha_i)$, ${x_1}^* \geq {x_2}^*$. 
\end{proof}

\section{Nash Equilibrium in the Ultimatum Game} \label{AppB}
We compute the set of symmetric Nash equilibria in the UG when $\alpha_1 = \alpha_2 = \alpha > 0$ and $\kappa_1 = \kappa_2 = \kappa > 0$. This extends \cite{alger2012homo} that do so when $\alpha = 0$ and $\kappa = 0$.
Let $x = (x_1, x_2)$ and $y = (y_1, y_2)$ denote the strategies of individuals $1$ and $2$, respectively. We normalize $v(0) \equiv 0$. Given that we focus on symmetric strategies $x_1 = y_1$ and $x_2 = y_2$, we distinguish when $x_1 > \frac{w}{2}$ or when $x_1 < \frac{w}{2}$. In the former, it is the \textit{proposer} who suffers from behindness aversion, while in the latter it is the \textit{responder}. When $x_1 \geq \frac{w}{2}$, the utility of individual $1$'s is given by: 
\begin{eqnarray}  \label{Case2}
u(x, y) &=&
\mathbf{1}_{\left\{x_1 \geq y_2\right\}} \left[(1 + \alpha - \kappa)v\left(w-x_1\right) - \alpha v(x_1)\right] \\
&+& \mathbf{1}_{\left\{y_1 \geq x_2\right\}} (1 - \kappa) v(y_1) \nonumber \\
&+&  \mathbf{1}_{\left\{x_1 \geq x_2\right\}} \kappa\left[v\left(w-x_1\right) + v\left(x_1\right)\right]. \nonumber
\end{eqnarray}

When $x_1 \leq \frac{w}{2}$, individual $1$'s utility is given by:
\begin{eqnarray} \label{Case1}
u(x, y) &=&
\mathbf{1}_{\left\{x_1 \geq y_2\right\}} (1 - \kappa)v\left(w-x_1\right) \\
&+& \mathbf{1}_{\left\{y_1 \geq x_2\right\}} \left[(1 + \alpha - \kappa)v\left(y_1\right) - \alpha v(w - y_1)\right] \nonumber \\
&+&  \mathbf{1}_{\left\{x_1 \geq x_2\right\}} \kappa \left[v\left(w-x_1\right) + v\left(x_1\right)\right]. \nonumber
\end{eqnarray}

\noindent \cite{alger2012homo} characterize the Nash equilibrium set when $\alpha = 0$ and $\kappa > 0$:
\begin{prop}
(in \Citealp{alger2012homo}): The set of homo-moralis strategies in this ultimatum-bargaining game, $X_\kappa$, is of the form: 
\begin{equation}
    X_\kappa=\left\{x \in[0,w]^2: \text { either } x_2 \leq x_1=\tau(\kappa) \text { or } \tau(\kappa)<x_1=x_2 \leq \rho(\kappa)\right\}
\end{equation}
where $\tau(\kappa) \leq \frac{w}{2}$ is defined by $\tau(\kappa)=\min \left\{x \in[0,w]: v^{\prime}(w - x) \geq \kappa v^{\prime}(x)\right\}$ and $\rho(\kappa) \geq \frac{w}{2}$ is continuous and decreasing in $\kappa$.
\end{prop}
 \par 
Therefore, there exist lower and upper bounds on the offers that can be sustained in equilibrium. When $x_1 > \tau(\kappa)$, any symmetric Nash equilibria must satisfy $x_1 = x_2$. On the other hand, when $x_1 = \tau(\kappa)$, any strategy $x = (x_1, x_2)$ with $x_1 = \tau(\kappa)$ and $x_2 \in [0, \tau(\kappa)]$ is a Nash equilibrium.\footnote{Therefore, the larger $\kappa$, the smaller the Nash equilibria set, similar to \cite{juan2024moral}.} \par 
We extend this result by considering the case with $\alpha > 0$. We start by restricting attention to equilibria where offers and rejection thresholds are not higher than $\frac{w}{2}$ (i.e., $x_1 \leq \frac{w}{2}$). In this case, spite only enters individuals' maximization problem in the \textit{responder} role (see \ref{Case1}). When incorporating $\alpha > 0$, individuals do not necessarily want to accept all positive offers, as some offers may give them a negative payoff. Let $\underline{x_2}(\kappa, \alpha) = \min\{x \in [0, w]$ : $v(x)(1 + \alpha - \kappa) \geq \alpha v(w - x)\}$ (this is equivalent to the definition in the main text).  Lemma \ref{Proof1} shows that there cannot exist a Nash equilibrium with ${x_1}^* < \underline{x_2}(\kappa, \alpha)$.
\begin{lemma} \label{Proof1}
    Let $x^*$ be a symmetric Nash equilibrium. Then, ${x_1}^* \geq \underline{x_2}(\kappa, \alpha)$.
\end{lemma}
\begin{proof}
    By contradiction, suppose there exists $x^*$ with ${x_1}^* < \underline{x_2} (\kappa, \alpha) < \frac{w}{2}$. We distinguish two cases depending on whether ${x_1}^* > {x_2}^*$ or if ${x_1}^* = {x_2}^* = x^*$. For the former, individual $1$ is better off simply increasing his rejection threshold, as this does not affect the first and third term of equation (\ref{Case2}), but increases the second one. For the latter, if $\kappa = 0$, then the individual sets ${x_2}^* = \underline{x_2} (\kappa, \alpha)$, given that he does not experience disutility when setting $x_1 < x_2$. When $\kappa > 0$, individual $1$ can profitably deviate by selecting $x_1 = x_2 = x^* + \varepsilon$, where $\varepsilon > 0$. That is, the individual slightly increases their offer and rejection threshold. This improves individual $1$'s utility, as while the second term of equation (\ref{Case2}) goes from negative to zero (as the individual 1 rejects individual $2$'s offer), the first and third term only change marginally. For the continuity of $v$, we can always find an $\varepsilon > 0$ sufficiently small such that the improvement of not accepting an offer below $\underline{x_2}(\kappa, \alpha)$ compensates for the marginal changes in the first and third terms.   
\end{proof}
We distinguish between two cases depending on whether (i) $\underline{x_2}(\kappa, \alpha) > \tau(\kappa)$ and (ii) $\underline{x_2}(\kappa, \alpha) \leq \tau(\kappa)$. When $\underline{x_2}(\kappa, \alpha) \leq \tau(\kappa)$, we have the same Nash equilibrium set characterized in \cite{alger2012homo}. The reason is that in such set the lowest offer possible, $\tau(\kappa)$, gives a positive material payoff to the responder. Hence, the individual does not have incentives to deviate from $({x_1}^*, {x_2}^*) = (\tau(\kappa), \underline{x_2}(\kappa, \alpha))$. \par 
On the other hand, when $\underline{x_2}(\kappa, \alpha) > \tau(\kappa)$, we have that the set of (symmetric) Nash equilibria is given by:
\begin{equation}
    X_{\{\kappa, \alpha\}} = \{x \in[0,\frac{w}{2}]^2: \underline{x_2}(\kappa, \alpha) \leq  x_2 = x_1 \leq \frac{w}{2}\}
\end{equation}  \par 
The reason is that when $\underline{x_2}(\kappa, \alpha) > \tau(\kappa)$, we cannot have a symmetric equilibrium with $x_2 < x_1$ (Claim 4 of \Citealp{alger2012homo}). Intuitively, if individuals were to set a rejection threshold strictly below $x_1$, they would be better off by decreasing their offer. \par 
We know consider offers above $\frac{w}{2}$. In that case, let $\overline{x_1}(\kappa, \alpha) = \max\{x \in [0, w]$ : $(1 - \kappa + \alpha)v(w - x) \geq v(x)\}$.

Lemma \ref{Proof2} shows that there cannot exist a (symmetric) Nash equilibrium with ${x_1}^* > \overline{x_1}(\kappa, \alpha)$.
\begin{lemma} \label{Proof2}
    Let $x^*$ be a symmetric Nash equilibrium. Then, ${x_1}^* \leq \overline{x_1}(\kappa, \alpha)$.
\end{lemma}
\begin{proof}
    By contradiction, suppose there exists $x^*$ with ${x_1}^* = {x_2}^* > \overline{x_1}(\kappa, \alpha)$. In that case, the first term in (\ref{Case2}) is negative, while the second and third terms are positive. In that case, individual is better off by selecting, e.g., ${x_1}^* = {x_2}^* = \frac{w}{2}$. This increases the first term in (\ref{Case2}) (it becomes zero instead of negative), does not affect the second term, and increases the their term. 
\end{proof}
 \noindent Proposition \ref{Papp} gives the set of Nash Equilibrium.
\begin{prop} \label{Papp}
The set of (symmetric) Nash Equilibrium in the UG when individuals have both $\alpha > 0$ and $\kappa > 0$ is of the form:
\begin{eqnarray}
X_{\{\kappa, \alpha\}} &=&
\{x \in[0,w]^2: \max[\underline{x_2}(\kappa, \alpha), \tau(\kappa)] \leq  x_2 = x_1 \leq \min[\overline{x_1}(\kappa, \alpha), \rho(\kappa)] \\
&& \text{ and } x_2 \geq x_1 = \tau(\kappa) \text{ when } \tau(\kappa) > \underline{x_2}(\kappa, \alpha) \}\nonumber
\end{eqnarray}     
\end{prop}
\begin{proof}
    Follows from combining the previous lemmas.
\end{proof}
Thus, the set of (symmetric) Nash equilibria shrinks when $\kappa$ and $\alpha$ increase, and hence the lowest offers and rejection thresholds that can be sustained in equilibrium are increasing in $\kappa$ and $\alpha$. Figure \ref{F2} displays the set of (symmetric) Nash equilibria when (i) $\alpha = 0$ (in blue) and (ii) $\alpha > 0$ and $\underline{x_2}(\kappa, \alpha) > \tau(\kappa)$ (in red). When $\alpha = 0$, we have that the equilibrium can be either (i) ${x_1}^* = {x_2}^* \in [\tau(\kappa), \rho(\kappa)]$, or (ii) ${x_1}^* = \rho(\kappa)$ and ${x_2}^* \in [0, \rho(\kappa)]$. On the other hand, when $\alpha > 0$ and $\underline{x_2}(\kappa, \alpha) > \tau(\kappa)$, then we have that the equilibrium is ${x_1}^* = {x_2}^* \in [\underline{x_2}(\kappa, \alpha), \overline{x_1}(\kappa, \alpha)]$. 
\begin{figure}[H]
  \centering
  \includegraphics[width=0.8\textwidth]{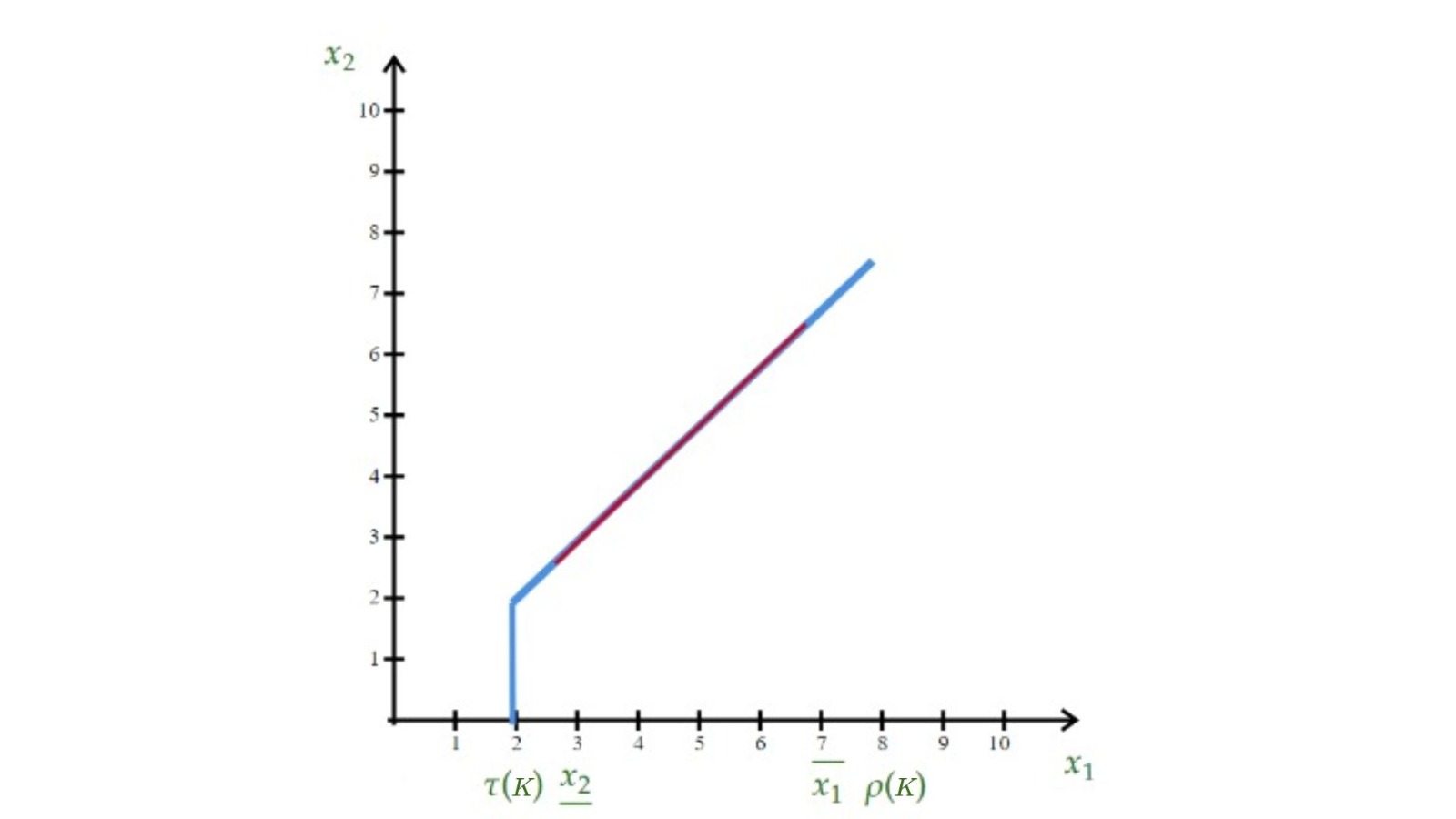}
  \caption{Set of (symmetric) Nash equilibria when (i) $\alpha = 0$ (in blue) and (ii) $\alpha > 0$ and $\underline{x_2}(\kappa, \alpha) > \tau(\kappa)$ (in red). $\tau(\kappa) = 2$, $\underline{x_2}(\kappa, \alpha) = 3$, $\overline{x_1}(\kappa, \alpha) = 7$, and $\rho(\kappa) = 8$.}
  \label{F2}
\end{figure}

\section{Simulations and Figures} \label{AppC}
\begin{figure}[H]
    \begin{center}
    \includegraphics[width= 0.7\textwidth]{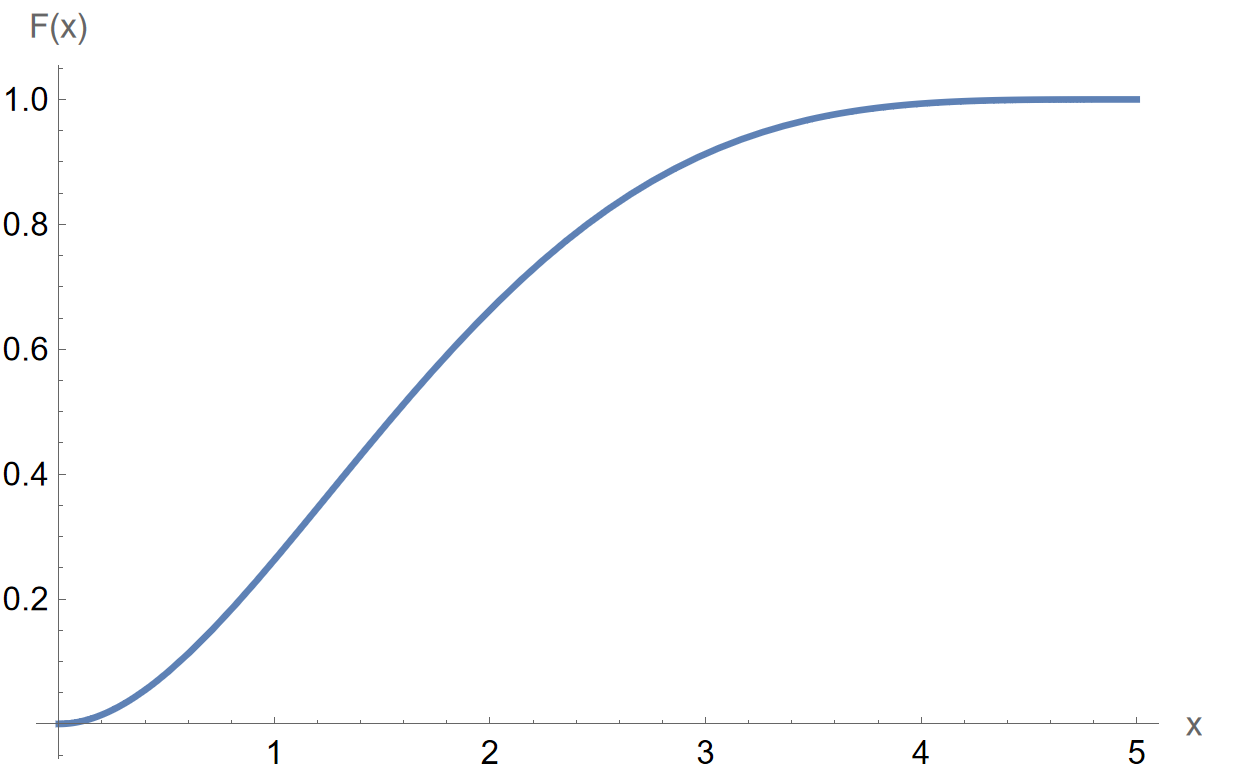}
     \captionsetup{labelformat=empty}
    \end{center}
    \caption{Distribution offers and rejection thresholds $F(x) = F_{[0,w/2]}(2, 4)$.}
    \label{Fig8}
\end{figure}

\begin{figure}[H]
    \centering
    \includegraphics[width=0.5\linewidth]{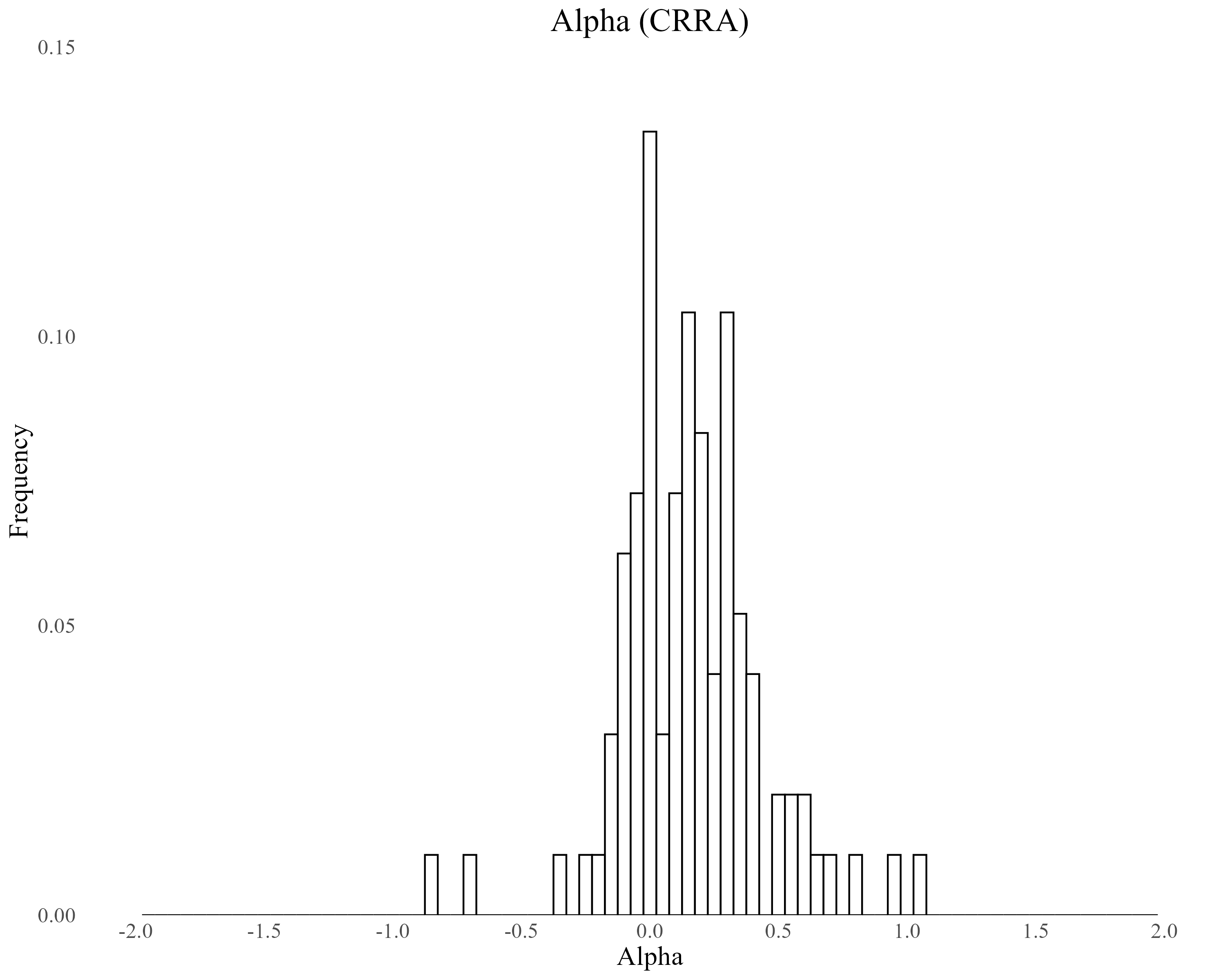}
    \caption{Distribution of individual estimates of $\alpha_i$ in \cite{van2023estimating}, CRRA estimates, for $\alpha, \beta \in [-2,2]$ and $\kappa \in [0,1]$.}
    \label{fig:alpha_hist_crra}
\end{figure}

\begin{figure}[H]
    \centering
    \includegraphics[width=0.5\linewidth]{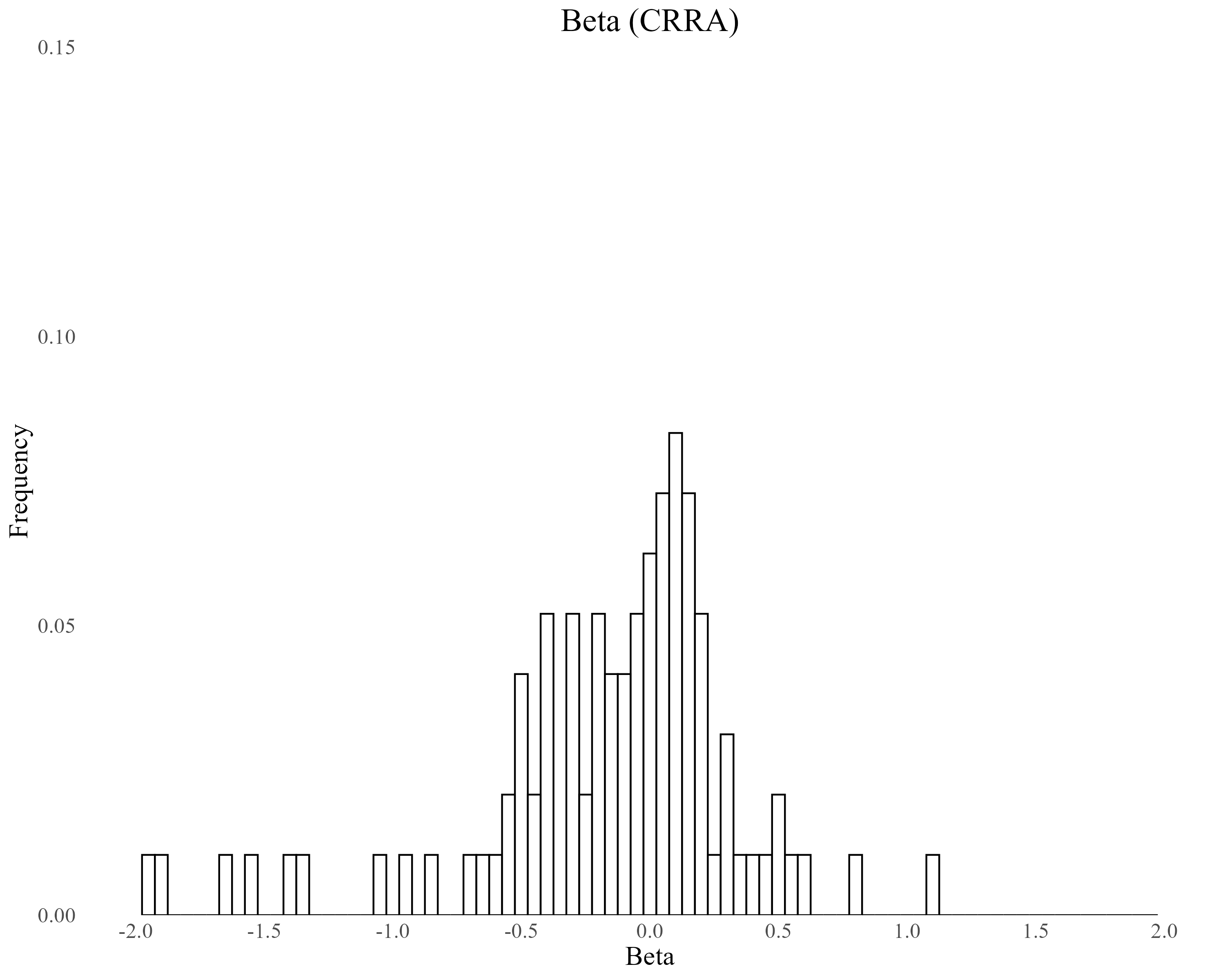}
    \caption{Distribution of individual estimates of $\beta_i$ in \cite{van2023estimating}, CRRA estimates, for $\alpha, \beta \in [-2,2]$ and $\kappa \in [0,1]$.}
    \label{fig:beta_hist_crra}
\end{figure}

\begin{figure}[H]
    \centering
    \includegraphics[width=0.5\linewidth]{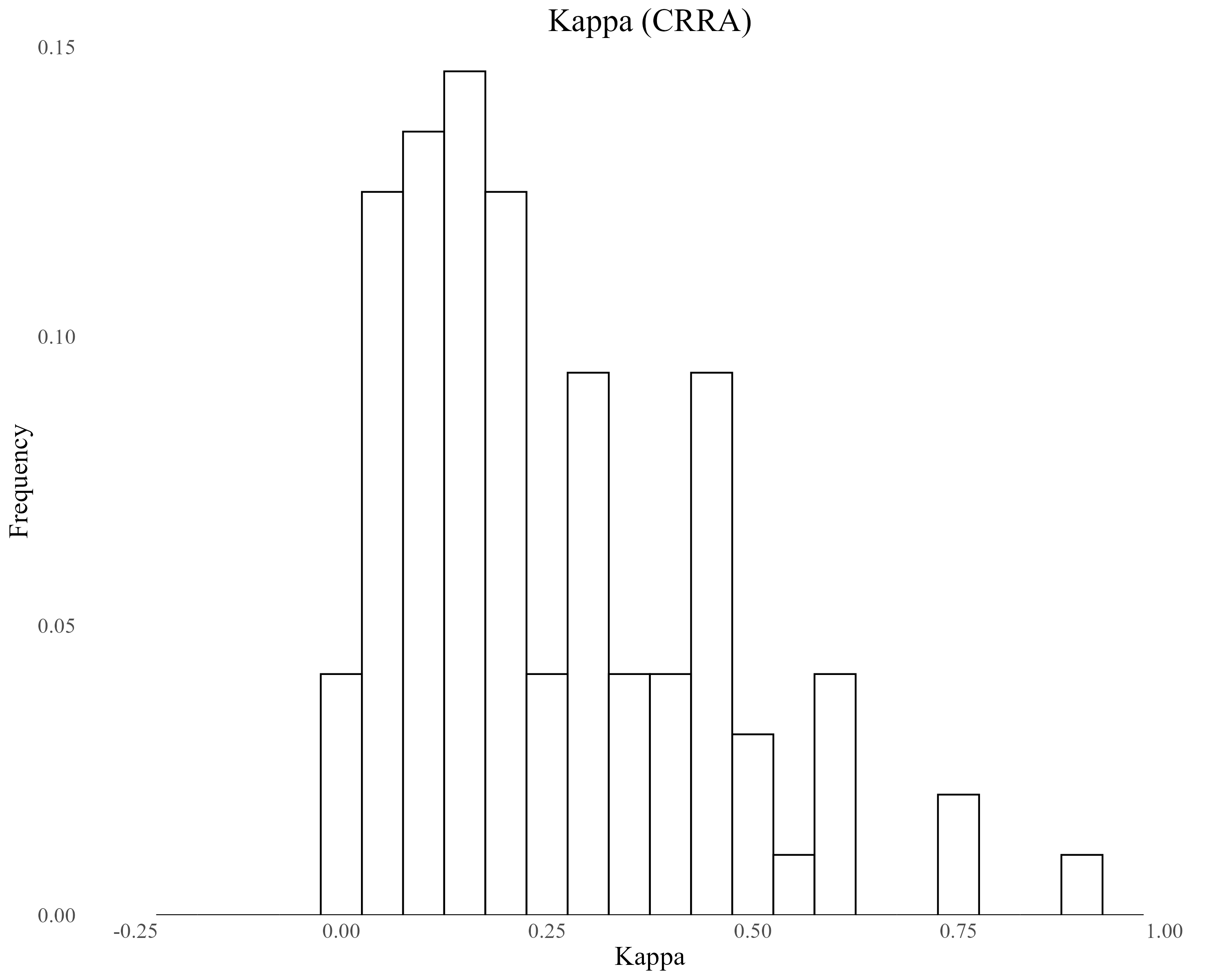}
    \caption{Distribution of individual estimates of $\kappa_i$ in \cite{van2023estimating}, CRRA estimates, for $\alpha, \beta \in [-2,2]$ and $\kappa \in [0,1]$.}
    \label{fig:kappa_hist_crra}
\end{figure}

\begin{figure}[H]
    \centering
    \includegraphics[width=0.5\linewidth]{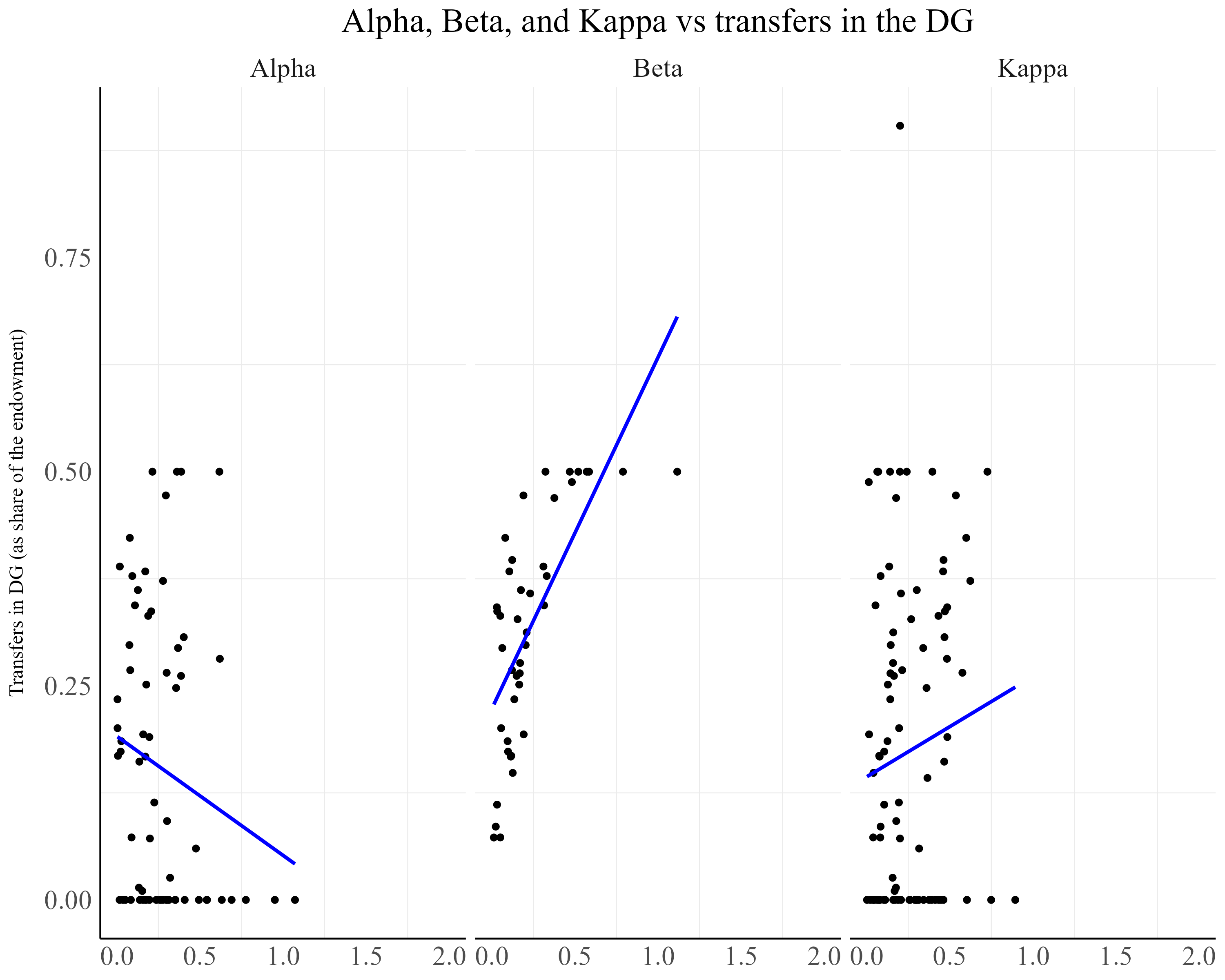}
    \caption{Transfers in the DG and parameter estimates. \cite{van2023estimating} CRRA estimates, for $\alpha, \beta \in [-2,2]$ and $\kappa \in [0,1]$.}
    \label{fig:scatterDG_w}
\end{figure}
\begin{figure}[H]
    \centering
    \includegraphics[width=0.75\linewidth]{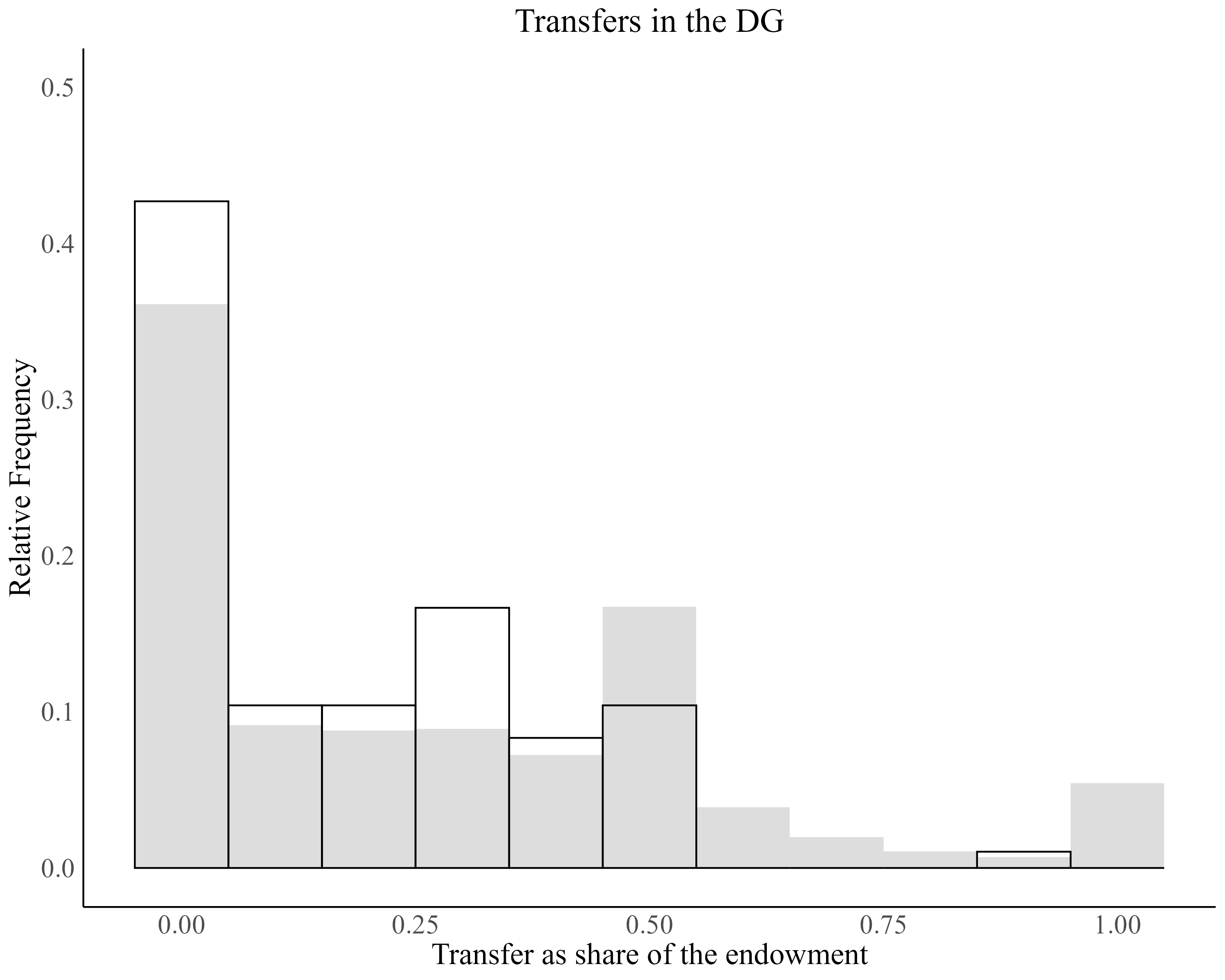}
    \caption{DG transfers reported in \cite{engel2011dictator} --in gray-- and transfers based on  \cite{van2023estimating} CRRA estimates, for $\alpha, \beta \in [-2,2]$ and $\kappa \in [0,1]$.}
    \label{fig:DG_engel}
\end{figure}

\begin{figure}[H]
    \centering
    \includegraphics[width=0.5\linewidth]{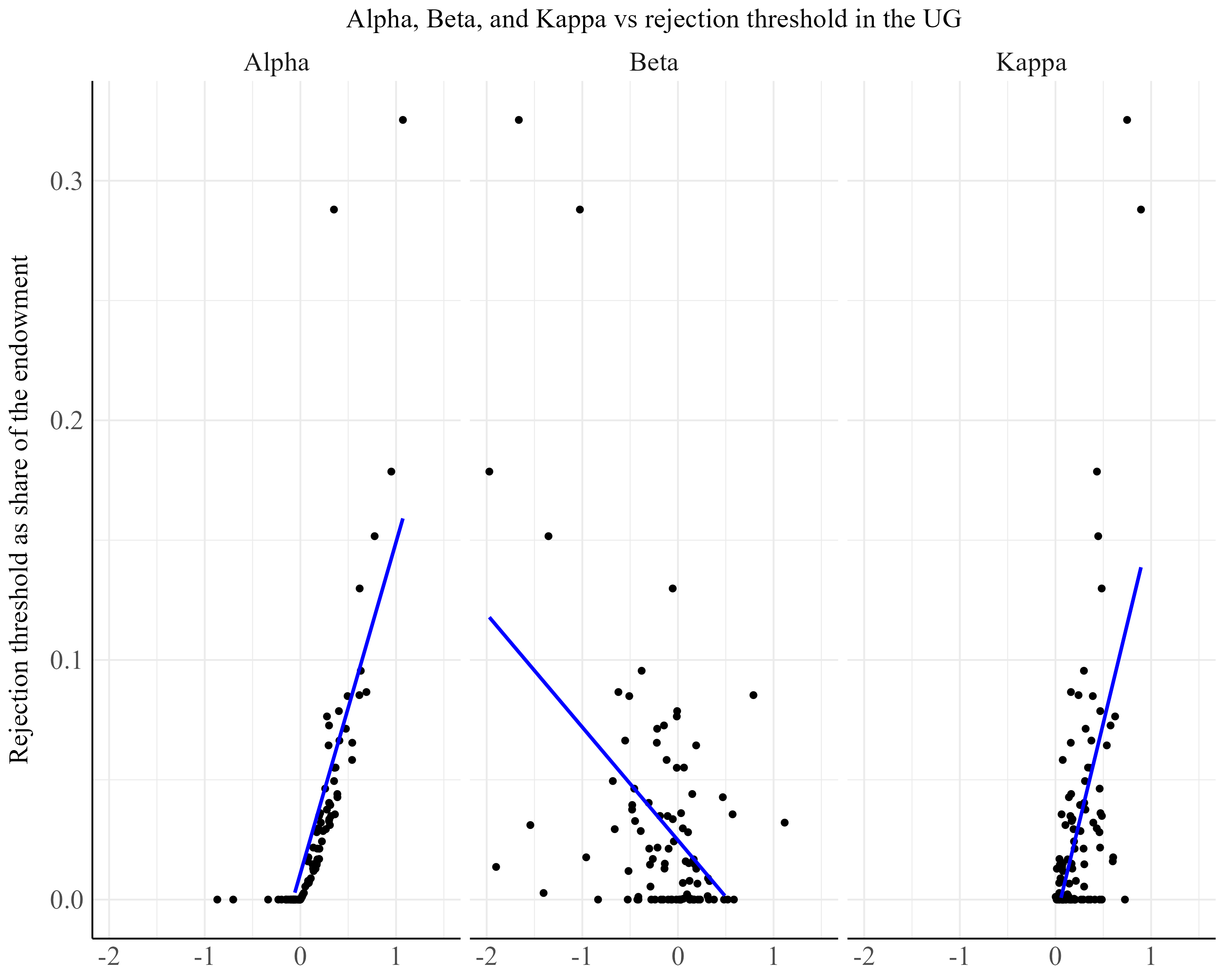}
    \caption{Rejection thresholds and parameter estimates. \cite{van2023estimating} CRRA estimates, for $\alpha, \beta \in [-2,2]$ and $\kappa \in [0,1]$.}
    \label{fig:scatterUG_w}
\end{figure}

\begin{figure}[H]
    \centering
    \includegraphics[width=0.8\linewidth]{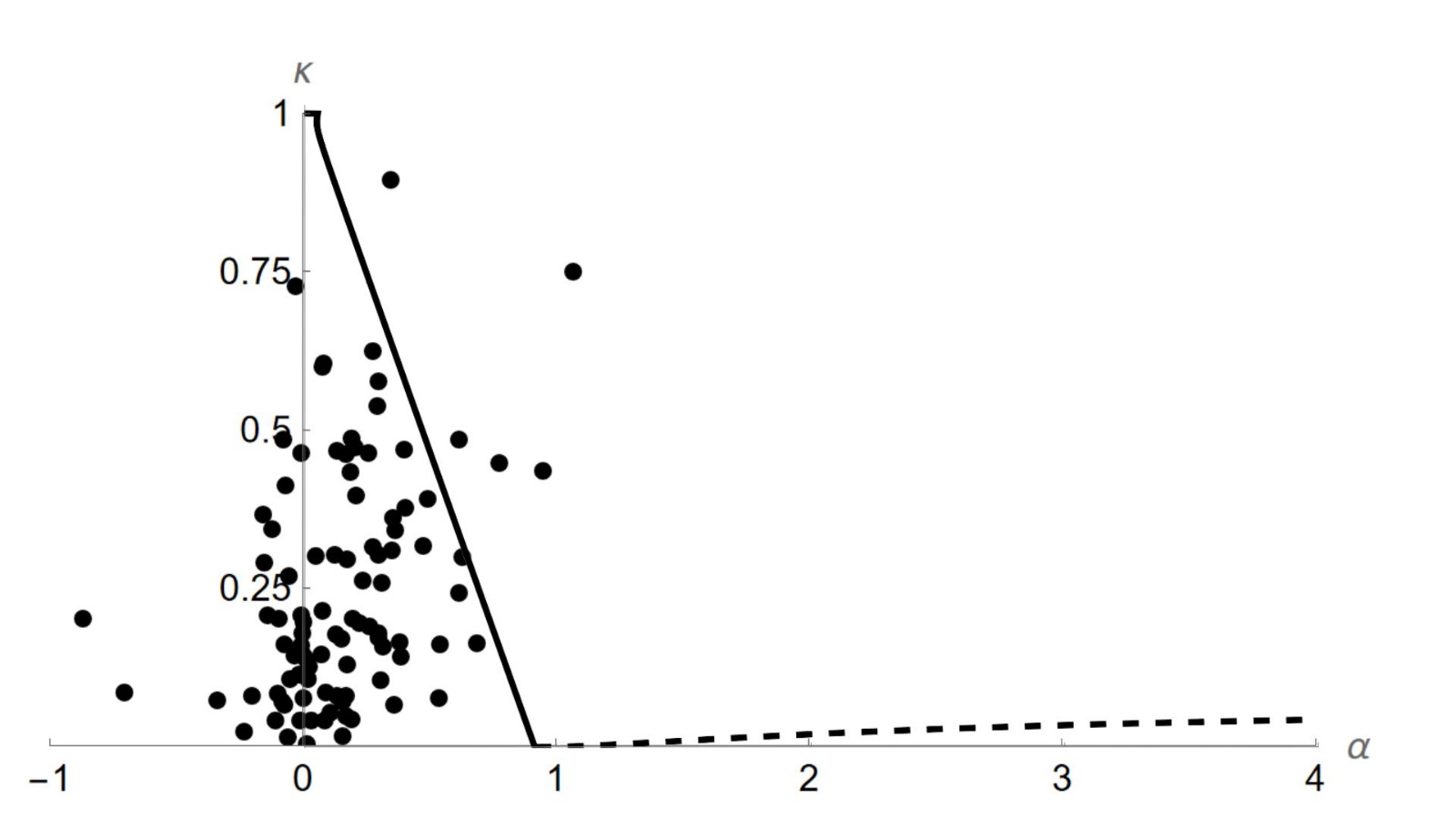}
    \caption{Optimal strategy in Figure 1 with individual estimates.}
    \label{fig:dots}
\end{figure}

          

\begin{table}[ht] \small
\centering
\begin{tabular}{lccccccc}
\toprule
Variable                              & Min         & Q1          & Median      & Q3           & Max          & Mean         & Obs. \\
\midrule
Donation DG                           & 0.00 (0.00) & 0.00 (0.00) & 6.62 (0.11) & 18.60 (0.32) & 53.20 (0.90) & 10.20 (0.17) & 96 \\
Threshold UG & 0.00 (0.00) & 0.00 (0.00) & 0.87 (0.01) & 2.34 (0.04)  & 19.10 (0.33) & 1.90 (0.03)  & 96 \\
\bottomrule
\end{tabular}
\caption{Descriptive Statistics of transfers in DG and Rejection Threshold in UG. \cite{van2023estimating} core sample, CRRA estimates, for $\alpha,\beta \in [-2,2]$ and $\kappa \in [0,1]$. Share of the endowment in parentheses.}
\label{tab:descriptive_stats_donation_DG_x2bar}
\end{table}

\begin{table}[ht] \small
\centering
\begin{tabular}{lccccccc}
\toprule
Variable                              & Min         & Q1          & Median      & Q3           & Max          & Mean         & Obs. \\
\midrule
Donation DG                           & 0.00 (0.00) & 0.00 (0.00) & 5.17 (0.09) & 19.90 (0.34) & 58.80 (1.00) & 10.40 (0.18) & 96 \\
Threshold UG & 0.00 (0.00) & 0.00 (0.00) & 0.00 (0.00) & 0.41 (0.01)  & 2.62 (0.04)  & 0.32 (0.01)  & 96 \\
\bottomrule
\end{tabular}
\caption{Descriptive Statistics of transfers in DG and Rejection Threshold in UG. \cite{van2023estimating} core sample, CRRA estimates with only $\alpha$ and $\beta$, for $\alpha,\beta \in [-2,2]$. Share of the endowment in parentheses.}
\label{tab:descriptive_stats_donation_DG_x2bar_ab}
\end{table}

\begin{figure}[H]
    \centering
    \includegraphics[width=0.7\linewidth]{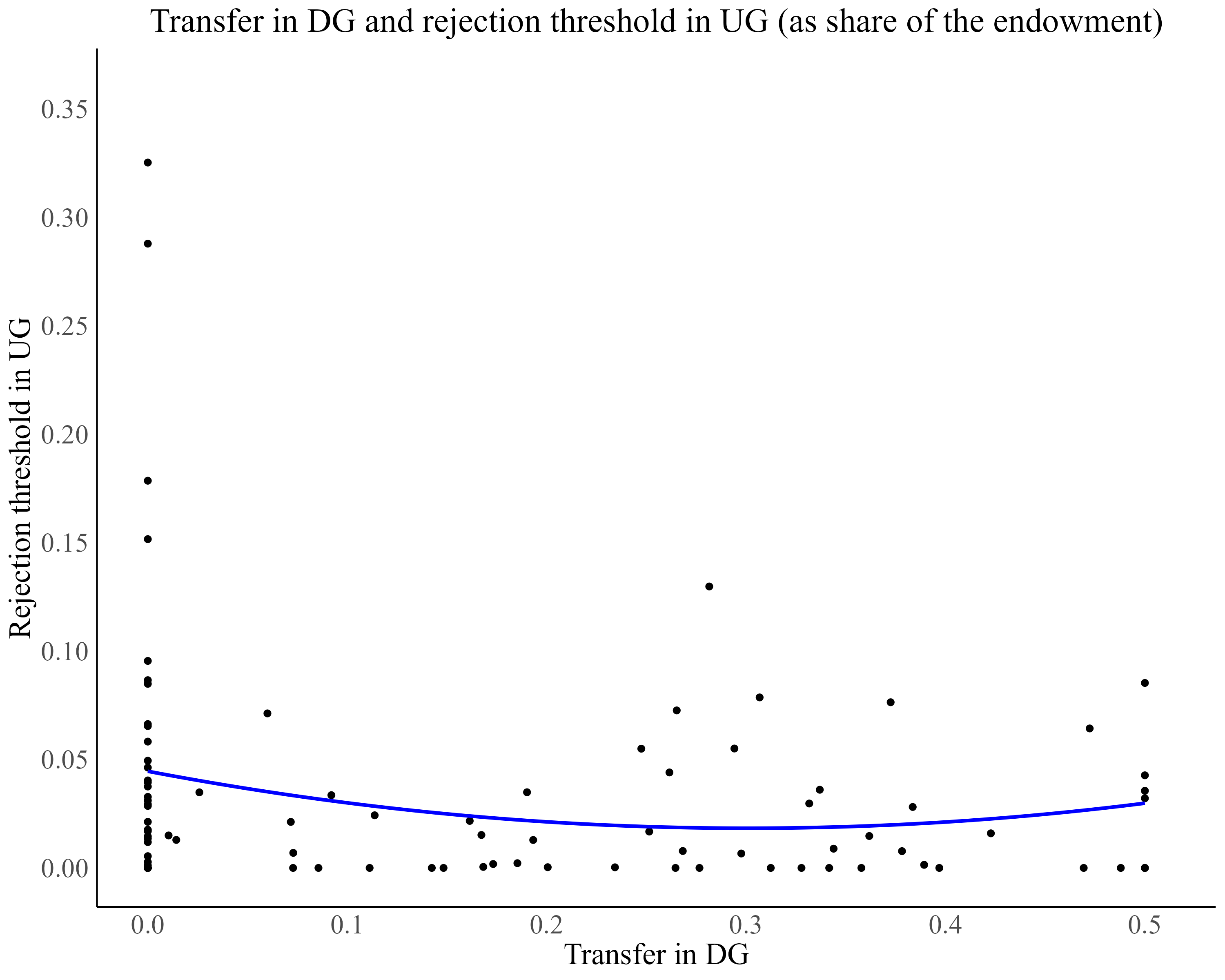}
    \caption{Transfer in DG and constrained optimal rejection threshold in UG. Both expressed as share of the endowment. \cite{van2023estimating} core sample, CRRA estimates, for $\alpha, \beta  \in [-2,2]$ and $\kappa \in [0,1]$. $N = 96$.}
    \label{fig:scatter_plot_w}
\end{figure}

\begin{figure}[H]
    \centering
    \includegraphics[width=0.6\linewidth]{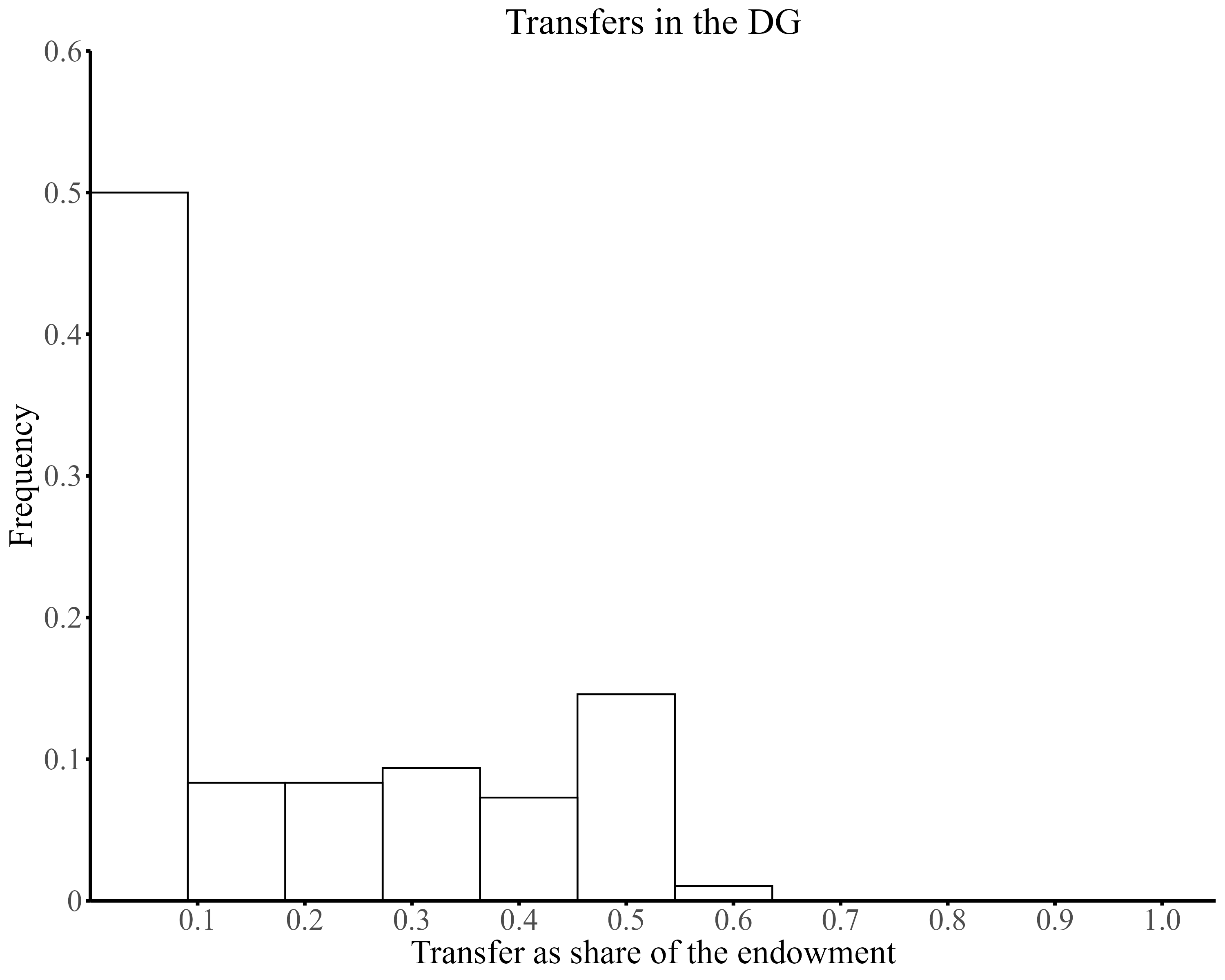}
    \caption{Predicted share transferred in the Dictator Game.  \cite{van2023estimating} CRRA estimates with only $\alpha$ and $\beta$; $\alpha, \beta \in [-2,2]$ and $w = 58.8$ (10 euros). $N = 96$.}  
    \label{fig:histDG_w_ab_all}
\end{figure}

\begin{figure}[H]
    \centering
    \includegraphics[width=0.6\linewidth]{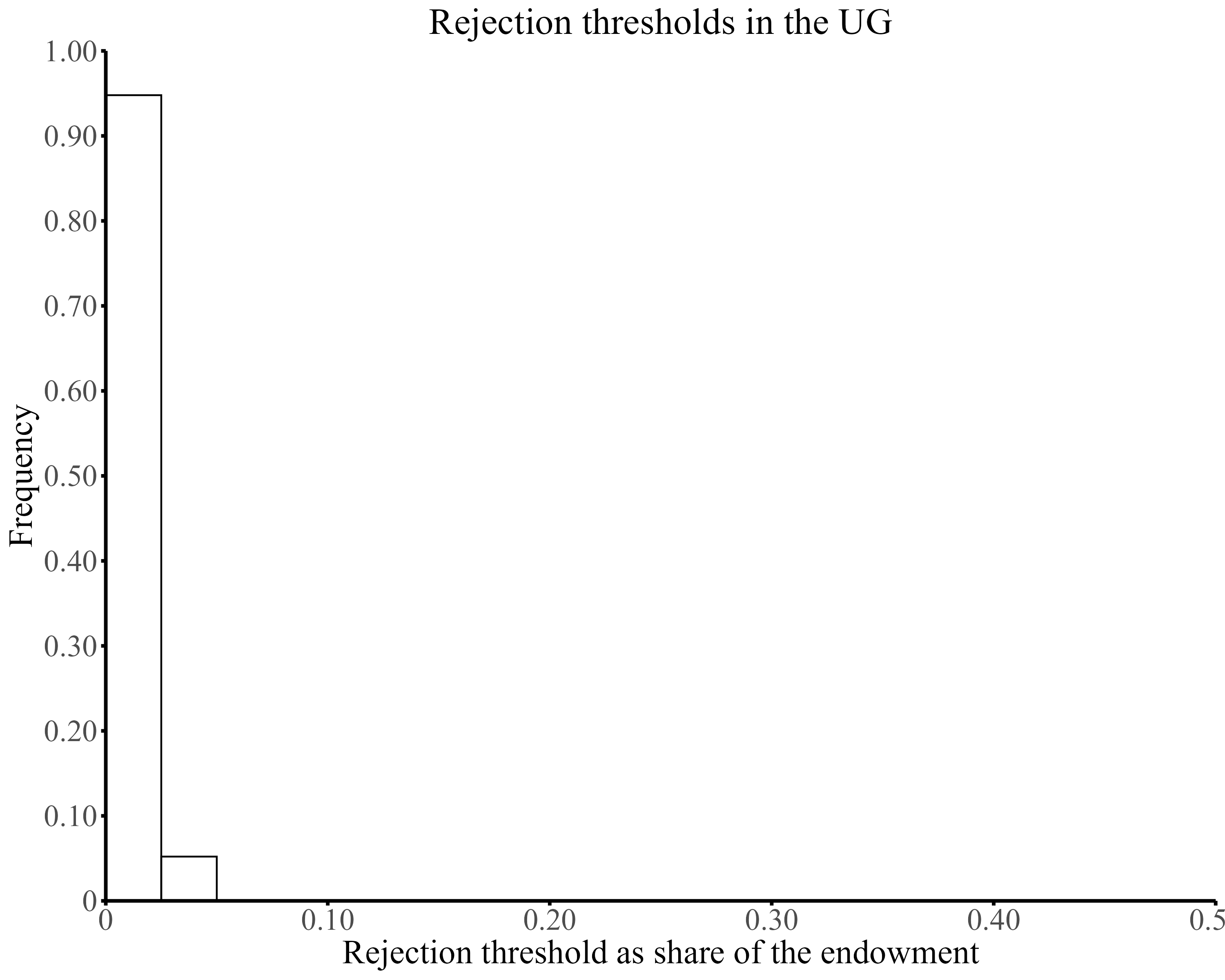}
    \caption{Constrained optimal rejection threshold $\underline{x_2}$ in the UG.\cite{van2023estimating} CRRA estimates with only $\alpha$ and $\beta$; $\alpha, \beta \in [-2,2]$ and $w = 58.8$ (10 euros). $N = 96$.}
    \label{fig:histUG_w_ab_all}
\end{figure}

\begin{figure}[H]
    \centering
    \includegraphics[width=0.6\linewidth]{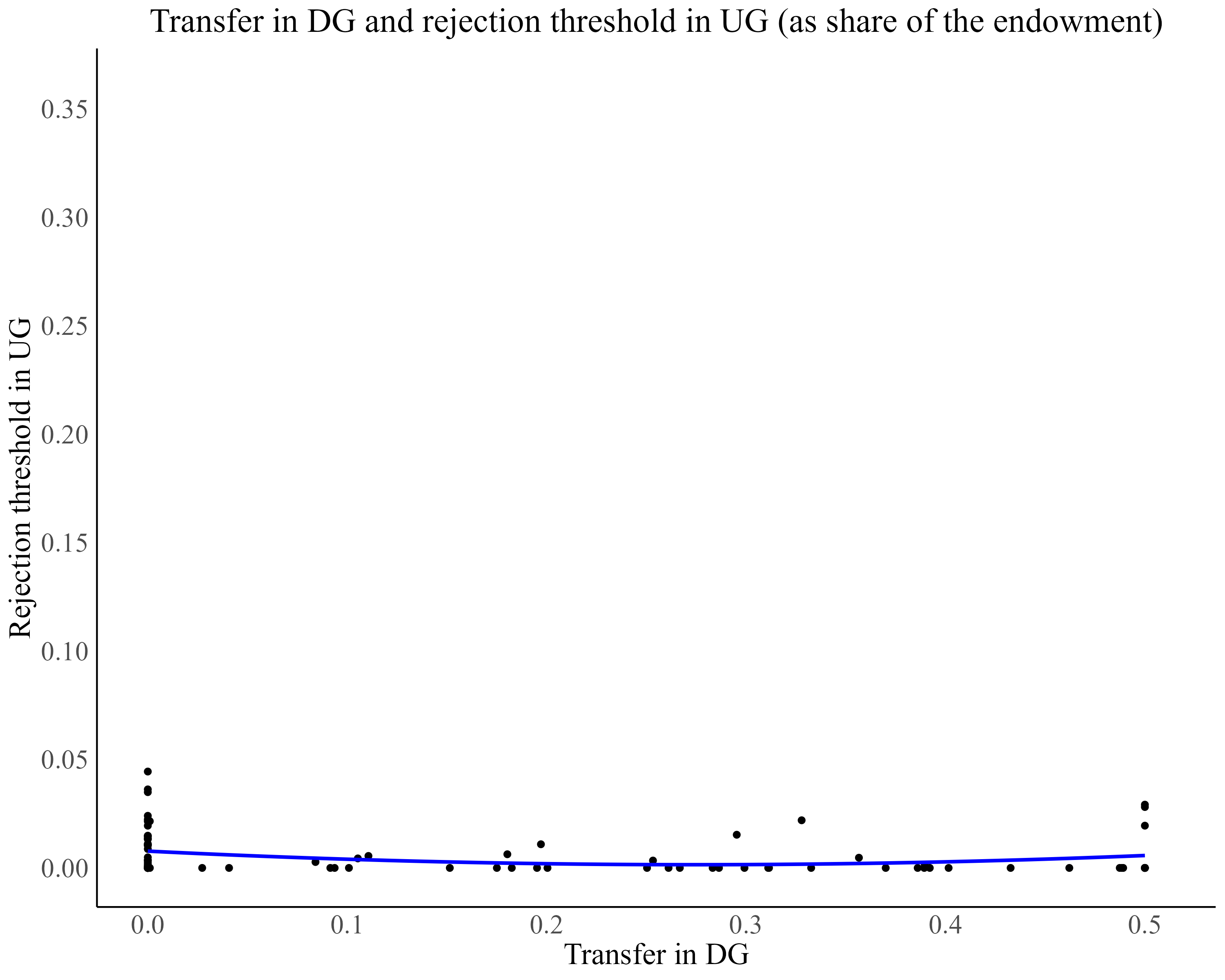}
    \caption{Transfer in DG and constrained optimal rejection threshold in UG. \cite{van2023estimating} CRRA estimates with only $\alpha$ and $\beta$; $\alpha, \beta \in [-2,2]$ and $w = 58.8$ (10 euros). $N = 96$.}
    \label{fig:scatter_plot_w_ab_all}
\end{figure}

\begin{figure}
    \centering
    \includegraphics[width=0.5\linewidth]{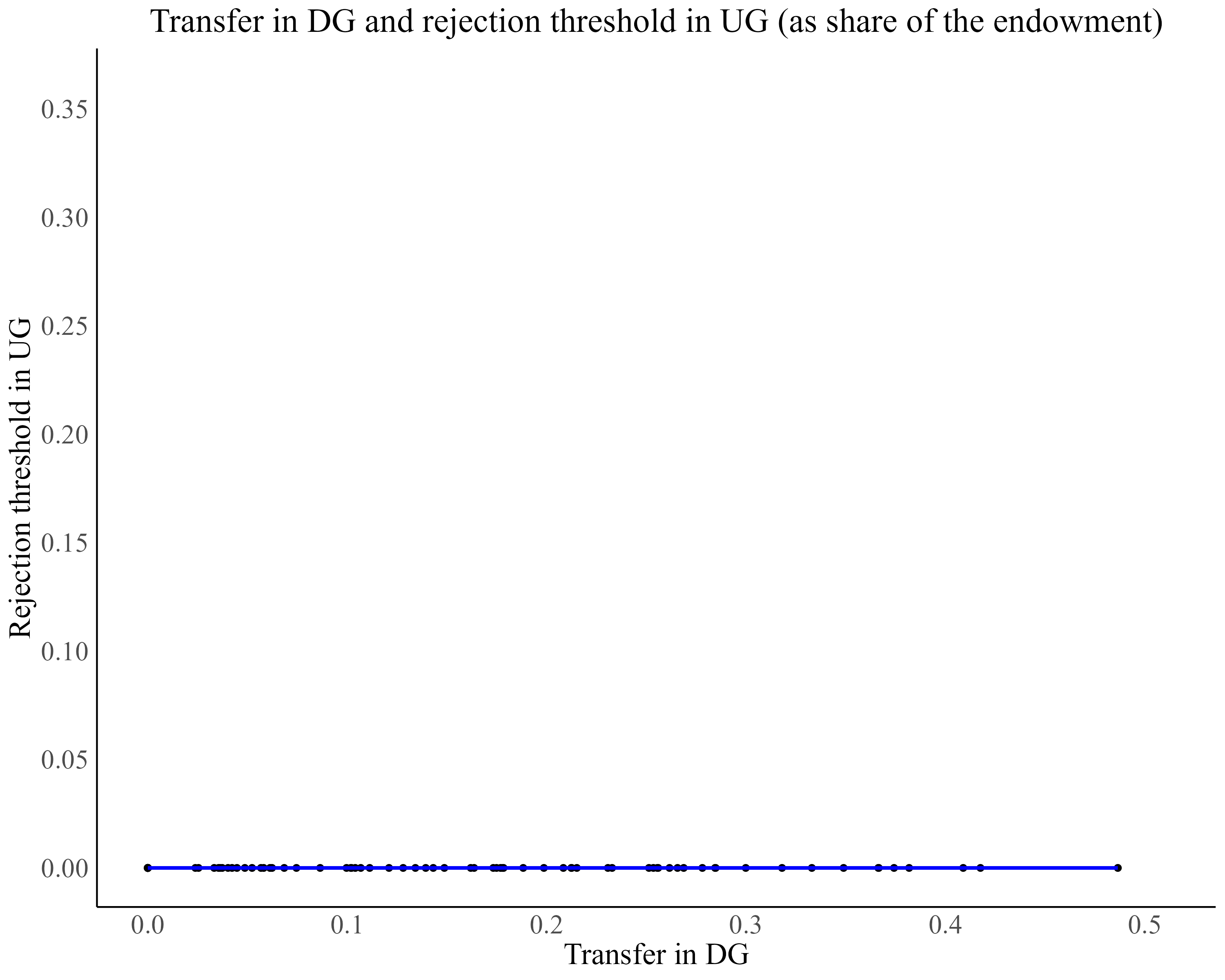}
    \caption{Transfer in DG and constrained optimal rejection threshold in UG. \cite{van2023estimating} CRRA estimates with only $\kappa\in [0,1]$, when $w = 58.8$ (10 euros). }
    \label{fig:scatter_plot_w_c_all}
\end{figure}

\section{Comparison with \cite{branas2014fair}}

\begin{figure}[H]
    \centering
    \includegraphics[width=0.5\linewidth]{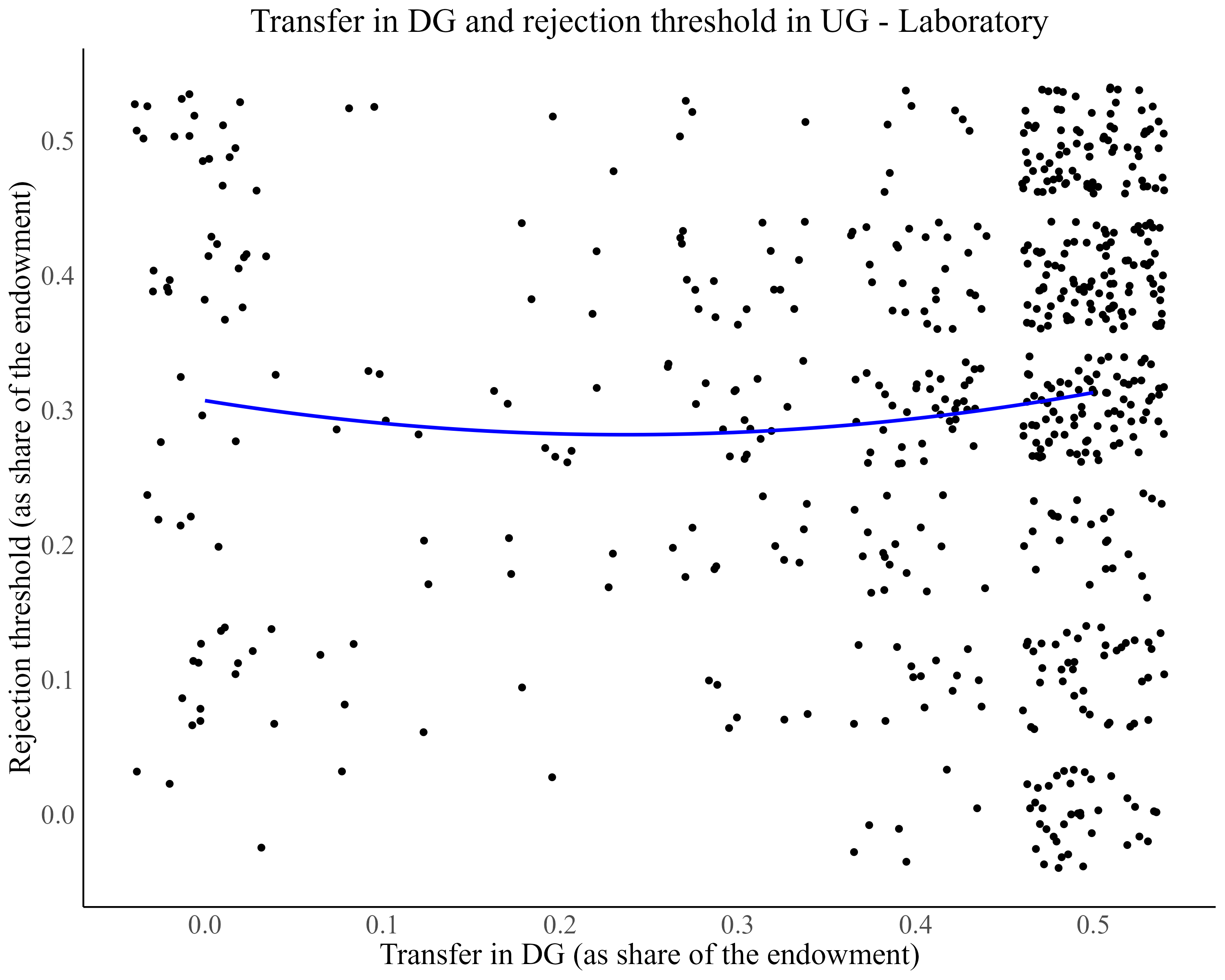}
    \caption{Transfer in DG and rejection threshold in UG in ``Laboratory'' study reported in \cite{branas2014fair}}
    \label{fig:scatter_plot_branas_lab}
\end{figure}

\begin{figure}[H]
    \centering
    \includegraphics[width=0.5\linewidth]{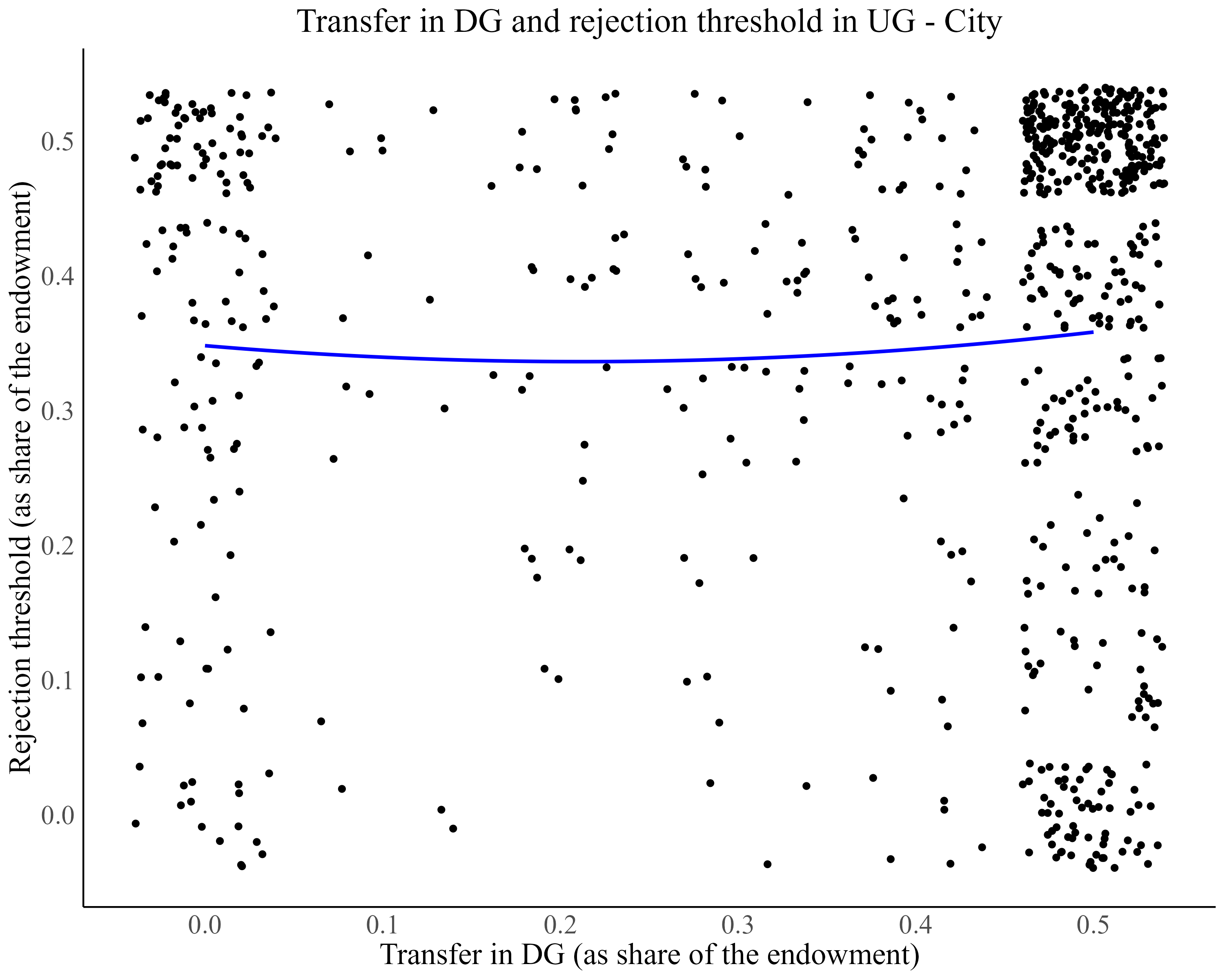}
    \caption{Transfer in DG and rejection threshold in UG in ``City'' study reported in \cite{branas2014fair}}
    \label{fig:scatter_plot_branas_city}
\end{figure}
\end{appendices}
\newpage

\bibliographystyle{chicago}
\bibliography{sample1}

\end{document}